\documentclass[lettersize,journal]{IEEEtran}
\usepackage{amsmath,amssymb,amsfonts}
\usepackage{xr}
\usepackage{pifont}
\usepackage{algorithmicx}
\usepackage[noend,ruled, linesnumbered, vlined]{algorithm2e}
\usepackage{graphicx}
\usepackage{epstopdf}
\usepackage{textcomp}
\usepackage{xcolor}
\usepackage{bm}
\usepackage{array}
\usepackage{ntheorem}
\usepackage{subfigure} 
\usepackage{cases}
\usepackage{stfloats}
\usepackage{url}
\usepackage{verbatim}
\usepackage{algpseudocode}  
\newenvironment{sequation}{\begin{equation}\small}{\end{equation}}

\newenvironment{proof}{{\noindent\it Proof. }}{\hfill $\blacksquare$\par}
\usepackage{enumitem}[topsep=3pt,itemsep=3pt]

\allowdisplaybreaks[4]
\def\BibTeX{{\rm B\kern-.05em{\sc i\kern-.025em b}\kern-.08em
    T\kern-.1667em\lower.7ex\hbox{E}\kern-.125emX}}
\definecolor{b}{rgb}{0.0, 0, 1}
\definecolor{r}{rgb}{1, 0, 0}
\ifCLASSOPTIONcompsoc
\usepackage[nocompress]{cite}
\else
\usepackage{cite}
\fi

\theoremseparator{.}
\newskip\theorempreskipamount
\newskip\theorempostskipamount
\newtheorem{theorem}{Theorem}
\newtheorem{definition}{Definition}
\newtheorem{lemma}{Lemma}
\theorembodyfont{}

\begin{document}

\title{QoE Maximization for Multiple-UAV-Assisted Multi-Access Edge Computing via An Online Joint Optimization Approach}

\author{
    Long~He,
    Geng~Sun,~\IEEEmembership{Senior Member,~IEEE},
    Zemin Sun, 
    Qingqing~Wu,~\IEEEmembership{Senior Member,~IEEE}, 
    Jiawen~Kang,~\IEEEmembership{Senior Member,~IEEE},
    Dusit Niyato,~\IEEEmembership{Fellow,~IEEE}, 
    Zhu~Han,~\IEEEmembership{Fellow,~IEEE}, and \\
    Victor C. M. Leung,~\IEEEmembership{Life Fellow,~IEEE}
  \thanks{This study is supported in part by the National Natural Science Foundation of China (62272194, 62471200), in part by the Science and Technology Development Plan Project of Jilin Province (20250101027JJ), in part by the U.S. National Science Foundation (NSF) (ECCS-2302469, CMMI-2222810), in part by Toyota, in part by Amazon, and in part by Japan Science and Technology Agency (JST) Adopting Sustainable Partnerships for Innovative Research Ecosystem (ASPIRE) (JPMJAP2326). (\textit{Corresponding author: Geng Sun and Zemin Sun.)}}
	\IEEEcompsocitemizethanks{    
 \IEEEcompsocthanksitem Long He, and Zemin Sun are with the College of Computer Science and Technology, Jilin University, Changchun 130012, China (e-mail: helong23@mails.jlu.edu.cn, sunzemin@jlu.edu.cn).
 	\IEEEcompsocthanksitem Geng Sun is with the College of Computer Science and Technology, Key Laboratory of Symbolic Computation and Knowledge Engineering of Ministry of Education, Jilin University, Changchun 130012, China, and also with the College of Computing and Data Science, Nanyang Technological University, Singapore 639798 (e-mail: sungeng@jlu.edu.cn).
    \IEEEcompsocthanksitem Qingqing Wu is with the Department of Electronic Engineering, Shanghai Jiao Tong University, Shanghai, China (e-mail: qingqingwu@sjtu.edu.cn). 
    \IEEEcompsocthanksitem Jiawen Kang is with the School of Automation, Guangdong University of Technology, Guangzhou 510006, China (e-mail: kjwx886@163.com).
    \IEEEcompsocthanksitem Dusit Niyato is with the College of Computing and Data Science, Nanyang Technological University, Singapore 639798 (e-mail: dniyato@ntu.edu.sg).
    \IEEEcompsocthanksitem Zhu Han is with the Department of Electrical and Computer Engineering at the University of Houston, Houston TX 77004, USA, and also with the Department of Computer Science and Engineering, Kyung Hee University, Seoul 446701, South Korea (e-mail: hanzhu22@gmail.com).
    \IEEEcompsocthanksitem Victor C. M. Leung is with the Artificial Intelligence Research Institute, Shenzhen MSU-BIT University, Shenzhen 518115, China, with the College of Computer Science and Software Engineering, Shenzhen University, Shenzhen 518060, China, and also with the Department of Electrical and Computer Engineering, The University of British Columbia, Vancouver V6T 1Z4, Canada (e-mail: vleung@ieee.org).}
	\thanks{A small part of this paper was published by IEEE INFOCOM 2024~\cite{he2024}.}}

\IEEEtitleabstractindextext{

%
%

\begin{abstract}
\par In disaster scenarios, conventional terrestrial multi-access edge computing (MEC) paradigms, which rely on ground infrastructure, may become unavailable due to infrastructure damage. With high-probability line-of-sight (LoS) communication, flexible mobility, and low cost, unmanned aerial vehicle (UAV)-assisted MEC is emerging as a promising paradigm to provide edge computing services for ground user devices (UDs) in disaster-stricken areas. However, the limited battery capacity, computing resources, and spectrum resources also pose serious challenges for UAV-assisted MEC, which can potentially shorten the service time of UAVs and degrade the quality of experience (QoE) of UDs without an effective control approach. To this end, in this work, we first present a hierarchical architecture of multiple-UAV-assisted MEC networks that enables the coordinated provision of edge computing services by multiple UAVs. Then, we formulate a joint task offloading, resource allocation, and UAV trajectory control optimization problem (JTRTOP) to maximize the QoE of UDs while considering the energy and resource constraints of UAVs. Since the problem is proven to be a future-dependent and NP-hard problem, we propose a novel online joint task offloading, resource allocation, and UAV trajectory control approach (OJTRTA) to solve the problem. Specifically, the JTRTOP is first transformed into a per-slot real-time optimization problem (PROP) using the Lyapunov optimization framework. Then, a two-stage optimization method based on game theory and convex optimization is proposed to solve the PROP. Simulation results show that the proposed OJTRTA outperforms various benchmark approaches and achieves at least a 10\% improvement in the QoE of UDs compared to deep reinforcement learning (DRL)-based algorithms, thereby validating the superiority of the proposed approach.
\end{abstract}

}

\maketitle
\IEEEdisplaynontitleabstractindextext
\IEEEpeerreviewmaketitle

%
%

\section{Introduction}
\label{sec:Introduction}

\par \IEEEPARstart{W}{ith} the development of wireless communication technologies and artificial intelligence, many intelligent applications with strict requirements on computing resources and latency have emerged rapidly, such as real-time video analysis, virtual and augmented realities, and interactive online games~\cite{shi2016edge}. However, the limited computing capabilities and battery capacity of user devices (UDs) make it difficult to maintain high-level quality of experience (QoE) for these intelligent applications~\cite{Hekmati}. To overcome this challenge, multi-access edge computing (MEC) has emerged as a promising paradigm to offer reliable edge computing services by deploying MEC servers at the network edge~\cite{Mao2017}. Specifically, UDs can offload computation-hungry and latency-sensitive tasks to nearby MEC servers for execution, which can effectively reduce task completion delay and energy consumption, thereby enhancing the QoE of UDs. However, conventional terrestrial MEC heavily relies on ground infrastructure, which can become unavailable in the event of a disaster where ground infrastructure may be damaged.

\par The limitations of conventional terrestrial MEC have prompted a paradigm shift toward unmanned aerial vehicle (UAV)-assisted MEC due to the high probability line-of-sight (LoS) communication, flexible mobility, and low cost~\cite{Li2023TMC}. Specifically, the high probability of LoS links for UAVs enhances the communication coverage, network capacity, and reliable connectivity. Moreover, the flexible mobility of UAVs enables rapid and on-demand deployment, which is crucial in disaster scenarios. Finally, the low cost makes the UAV-assisted MEC system feasible and scalable. Therefore, UAV-assisted MEC holds tremendous potential for providing edge computing services to UDs in disaster-stricken areas.

\par However, several fundamental challenges should be addressed to fully exploit the advantages of UAV-assisted MEC. \textbf{\textit{i) Resource Allocation.}} Various delay-sensitive and computation-intensive tasks of UDs impose stringent demands on both computing and communication resources, while also exhibiting heterogeneous and time-varying characteristics. However, the limited computing and spectrum resources of UAVs leads to the competition for resources \cite{Sun2025,10418937}. Consequently, it is challenging to fully utilize the limited resources of UAVs to meet the heterogeneous and stringent demands of UDs. \textbf{\textit{ii) Task Offloading.}} Due to resource competition, the offloading decision of each UD is influenced not only by its own offloading requirements but also by the offloading decisions of other UDs, resulting in coupled and complex offloading decision-making among UDs \cite{Gao2024,bai2022delay}. \textbf{\textit{iii) UAV Trajectory Control.}} While the mobility of UAVs augments the flexibility and elasticity of edge computing services, the limited battery capacity of UAVs inherently restricts the service time \cite{Nguyen2024,zhou2022two}. Therefore, the trajectory of UAVs needs to be effectively controlled to balance the service time of UAVs and the QoE of UDs, which presents an additional challenge.

\par The aforementioned challenges require an efficient joint optimization of resource allocation, task offloading, and UAV trajectory control. However, focusing on just one aspect of these components is insufficient due to the following reasons. On the one hand, due to the limited resources and energy of UAVs, these optimization variables are mutually coupled. For example, the positions of UAVs affect the transmission rate of tasks and the resource allocation influences the completion delay of tasks. Therefore, optimizing task offloading requires simultaneous consideration of resource allocation and the positions of UAVs. On the other hand, these optimization variables collectively determine the QoE of UDs. Therefore, these interconnected optimization variables should be jointly optimized to fully exploit the UAV-assisted MEC system performance, as it can effectively capture the intricate and coupling interactions and trade-offs among various optimization components. To this end, we present a novel online joint optimization approach for task offloading, resource allocation, and UAV trajectory control to maximize the QoE of UDs while adhering to the resource and energy constraints of UAVs. The key contributions of this work are outlined as follows:

\begin{itemize}
\item \textbf{\textit{System Architecture.}} We propose a hierarchical architecture for multiple-UAV-assisted MEC networks to efficiently coordinate multiple UAVs in providing aerial edge computing services for UDs. Furthermore, we take into account the dynamic mobility and time-varying computational demands of UDs to accurately capture the features of real-world edge computing scenarios within the proposed system.

\item \textbf{\textit{Problem Formulation.}} We formulate a joint task offloading, resource allocation, and UAV trajectory control optimization problem (JTRTOP), with the aim of maximizing the QoE of UDs under the energy and resource constraints of UAVs. Specifically, the QoE of UDs is theoretically quantified by synthesizing the UD energy consumption and task completion delay. Moreover, the optimization problem is proved to be future-dependent and NP-hard.
	
\item \textbf{\textit{Algorithm Design.}} To solve the JTRTOP, we propose a novel online joint task offloading, resource allocation, and UAV trajectory control approach (OJTRTA). Specifically, we first transform the JTRTOP into a per-slot real-time optimization problem (PROP) by using the Lyapunov optimization framework. Then, we propose a two-stage optimization method to jointly optimize task offloading, resource allocation, and UAV trajectory control within the PROP using convex optimization and game theory.

\item \textbf{\textit{Validation.}} The effectiveness of the proposed OJTRTA is validated through both theoretical analysis and simulation experiments. In particular, the theoretical analysis establishes that the proposed OJTRTA not only satisfies the energy consumption constraints of UAVs, but also has low computational complexity. Additionally, the simulation results affirm the superiority of OJTRTA.	
\end{itemize}

\par The subsequent sections of this work are structured as follows. In Section \ref{sec:Related Work}, we review and summarize the related work. Section \ref{sec:System Model and problem Formulation} details the relevant system models and problem formulation. Section \ref{sec:Lyapunov-Based problem Transformation} describes the Lyapunov-based problem transformation. Section \ref{sec:Two-Stage Optimization Algorithm} presents the two-stage optimization method and theoretical analysis. Section \ref{sec:Simulation Results and Analysis} presents and discusses the simulation results. Finally, this work is concluded in Section \ref{sec:Conclusion}.

%
%

\section{Related Work}
\label{sec:Related Work}

\par In this section, we provide a comprehensive review of the relevant studies pertaining to UAV-assisted MEC architecture, formulation of joint optimization problems, and online optimization approach. Moreover, we outline the distinctions between this work and the related works in Table I of the supplementary material.

\subsection{UAV-assisted MEC Architecture}
\label{subsec:UAV-enabled MEC Architecture}

\par UAV-assisted MEC is a promising paradigm to dynamically expand the computational capabilities of MEC networks and facilitating emergency scenarios. From the perspective of system architecture, existing studies can be broadly classified into single UAV-assisted MEC architecture and multiple-UAV-assisted MEC architecture.

\par For single UAV-assisted MEC networks, Lin et al.~\cite{lin2023pddqnlp} explored the maximization of UAV energy efficiency while simultaneously considering fairness in task offloading within a single UAV-assisted MEC network. Yu et al.~\cite{parameter_k} presented a single UAV-assisted MEC network, where a UAV is deployed to provide MEC services to Internet of Things (IoT) devices in areas that are inaccessible to the edge cloud due to ground signal blockage or shadowing. Pang et al.~\cite{Pang2021} proposed a novel mmWave-enabled non-orthogonal multiple access (NOMA) communication scheme, which significantly enhances the energy efficiency of UAV networks. However, a single UAV possesses insufficient coverage range and resources, making it challenging to provide effective services in large-scale areas or areas with a high density of users.

\par Multiple-UAV-assisted MEC networks, which harness the computational and communication capabilities of multiple UAVs, have been increasingly gaining attention. For example, Li et al.~\cite{li2023robust} investigated a multiple-UAV-assisted MEC system and devised a robust approach for optimizing computation offloading and trajectory control. Shidrokh et al.~\cite{goudarzi2023uav} conducted a study on resource allocation and task offloading in a multiple-UAV-assisted MEC network, and proposed a two-layer collaborative evolution model to reduce energy consumption and task computation time. Bai et al.~\cite{bai2022delay} proposed a multiple-UAV-enabled edge-cloud computing system that leverages the collaboration between UAVs and remote clouds to deliver exceptional computational capabilities. 

\par Although multiple-UAV architectures show more significant advantages compared to single-UAV architectures, effective control is required for the collaboration among multiple UAVs, which has not been extensively studied yet. In this paper, we propose a hierarchical multiple-UAV-assisted MEC architecture to address the shortcomings of existing research.

\subsection{Formulation of Joint Optimization Problems}
\label{subsec:Joint Task Offloading, Resource Allocation and Trajectory Planning}

\par The formulation of a joint optimization problem is crucial for enhancing the performance of UAV-assisted MEC networks due to the constraints imposed by multi-dimensional resources. Next, we will provide an overview of existing research from two perspectives, i.e., optimization objectives and optimization variables, while highlighting the novelty of this work.

\subsubsection{Optimization Objectives}
\label{subsubsec:Optimization Objectives}

\par Several studies focus on minimizing task completion latency. For instance, Luan et al.~\cite{LuanCZL22} presented a joint optimization problem for topology reconstruction and subtask scheduling to minimize the average completion time of subtasks in a UAV-assisted MEC system. Deng et al.~\cite{DengFW23} investigated the minimization of task service latency in air-ground integrated wireless networks, while considering constraints on learning accuracy and energy consumption. 

\par Several studies focus on minimizing the energy consumption. For example, in the context of a multiple-UAV-assisted two-stage MEC system, Ei et al.~\cite{EiATHH22} aimed to minimize the energy consumption of both mobile devices and UAVs under constraints of task tolerance latency and available resources. Liu et al.~\cite{LiuWZWH22} conducted a study on a multi-input single-output UAV-assisted MEC network and presented a system energy minimization problem. 

\par There are also the studies that consider the combination of task completion latency and energy consumption as the optimization objective. For example, Tong et al.~\cite{TongLMCZW23} introduced a UAV-enabled multi-hop cooperative paradigm and formulated an optimization problem to maximize the system utility. The system utility was constructed by jointly considering UAV energy consumption and task completion latency. Hao et al.~\cite{10381761} explored the problem of maximizing the long-term average system gain in a UAV-assisted MEC system, where the system gain was theoretically constructed by combining task completion latency and system energy consumption.

\par This work differs from the aforementioned studies in the following aspects. First, we considers a more comprehensive set of performance metrics, including task completion latency, UD energy consumption, and UAV energy consumption, which more accurately reflect the practical characteristics of UAV-assisted MEC networks. Furthermore, we combine task completion latency and UD energy consumption into a UD cost function as the optimization objective, while UAV energy consumption is treated as a constraint. This formulation can simultaneously enhance the QoE of UDs and ensure reliable service time for UAVs.

\subsubsection{Optimization Variables}
\label{subsubsec:Optimization Variables}

\par Researchers have extensively explored various aspects of UAV-assisted MEC systems, with particular emphasis on resource allocation, task offloading, and UAV trajectory control. For example, Chen et al.~\cite{ChenBHZLSZL22} conducted research on task offloading for IoT nodes and trajectory control for multiple UAVs to reduce UAV energy consumption. Song et al.~\cite{song2022evolutionary} focused on investigating task offloading and UAV trajectory control problems to improve the performance of UAV-assisted MEC systems. To minimize latency and energy consumption, Pervez et al.~\cite{pervez2023energy} jointly optimized task offloading, CPU frequency allocation, transmission power, and UAV trajectory. Seid et al.~\cite{seid2023multi} explored the joint optimization of task offloading, age of information, computation resource allocation, and communication resource allocation to enhance the performance of UAV-enabled Internet of Medical Things networks.

\par These previous studies have several limitations. First, resource allocation, task offloading, and UAV trajectory control are key factors in enhancing the performance of UAV-assisted MEC networks. However, the aforementioned studies did not comprehensively optimize these factors, which hinders the full exploitation of the advantages of UAV-assisted MEC. Furthermore, these studies consider resource allocation from either the communication or the computation aspects, which may cause performance degradation in practical UAV-assisted MEC systems where both communication and computing resources are insufficient. Motivated by these issues, we jointly optimize task offloading, communication and computing resource allocation, as well as UAV trajectory control.

\subsection{Online Optimization Approach}
\label{subsec:Online Optimization Approach}

\par To tackle the intricate joint optimization problems of task offloading, resource allocation, and trajectory control for UAV-assisted MEC systems, numerous works have employed offline approaches for system design. These offline approaches formulate solutions under the assumption that the user's location remains unchanged and the demands of users are either fixed or known in advance~\cite{Xu2021,UAV-H,Hu2019}. However, in various edge computing scenarios, such as online games and real-time video analysis, computation demands arrive in a stochastic manner, and user locations exhibit dynamic changes. Therefore, it is necessary to develop online approaches to make real-time decisions without knowing future information.

\par Several studies have also explored online approaches. For example, Zhou et al.~\cite{zhou2022two} developed a two time-scale online approach for caching and task offloading by leveraging the Lyapunov optimization framework. Considering the time-varying computing requirements of user equipment, Wang et al.~\cite{Wang2022} jointly optimized the user association, resource allocation and trajectory of UAVs with the aim of minimizing energy consumption of all user equipment. To minimize the average power consumption of the system with randomly arriving user tasks, Hoang et al.~\cite{10102429} developed a Lyapunov-guided deep reinforcement learning (DRL) framework. Miao et al.~\cite{10388042} proposed a
deep deterministic policy gradient (DDPG)-based algorithm to optimize the computational resources allocation and UAV flight trajectory for UAV-assisted MEC. Xu et al.~\cite{10418937} formulated a long-term optimization problem for the joint optimization of UAV trajectory and computation resource allocation and proposed a trajectory control algorithm based on proximal policy optimization (PPO).

\par The Lyapunov-based optimization framework and DRL represent two viable methodologies for developing online approaches. While DRL is a powerful technique for training agents to make real-time decisions, it necessitates a substantial amount of sample data to learn optimal strategies and suffers from a lack of interpretability. In contrast, the Lyapunov-based optimization framework does not depend on sample data and provides stable performance guarantees. Therefore, we employ the Lyapunov-based optimization framework to devise our online approach. In a departure from existing research, we propose a novel two-stage optimization method based on game theory and convex optimization within the Lyapunov-based optimization framework. This method demonstrates both low computational complexity and superior performance.

%
%
\section{System Model and Problem Formulation}
\label{sec:System Model and problem Formulation}

\begin{figure*}[!hbt] 
	\centering
	\setlength{\abovecaptionskip}{0pt}%
	\setlength{\belowcaptionskip}{0pt}%
	\includegraphics[width =6.5in]{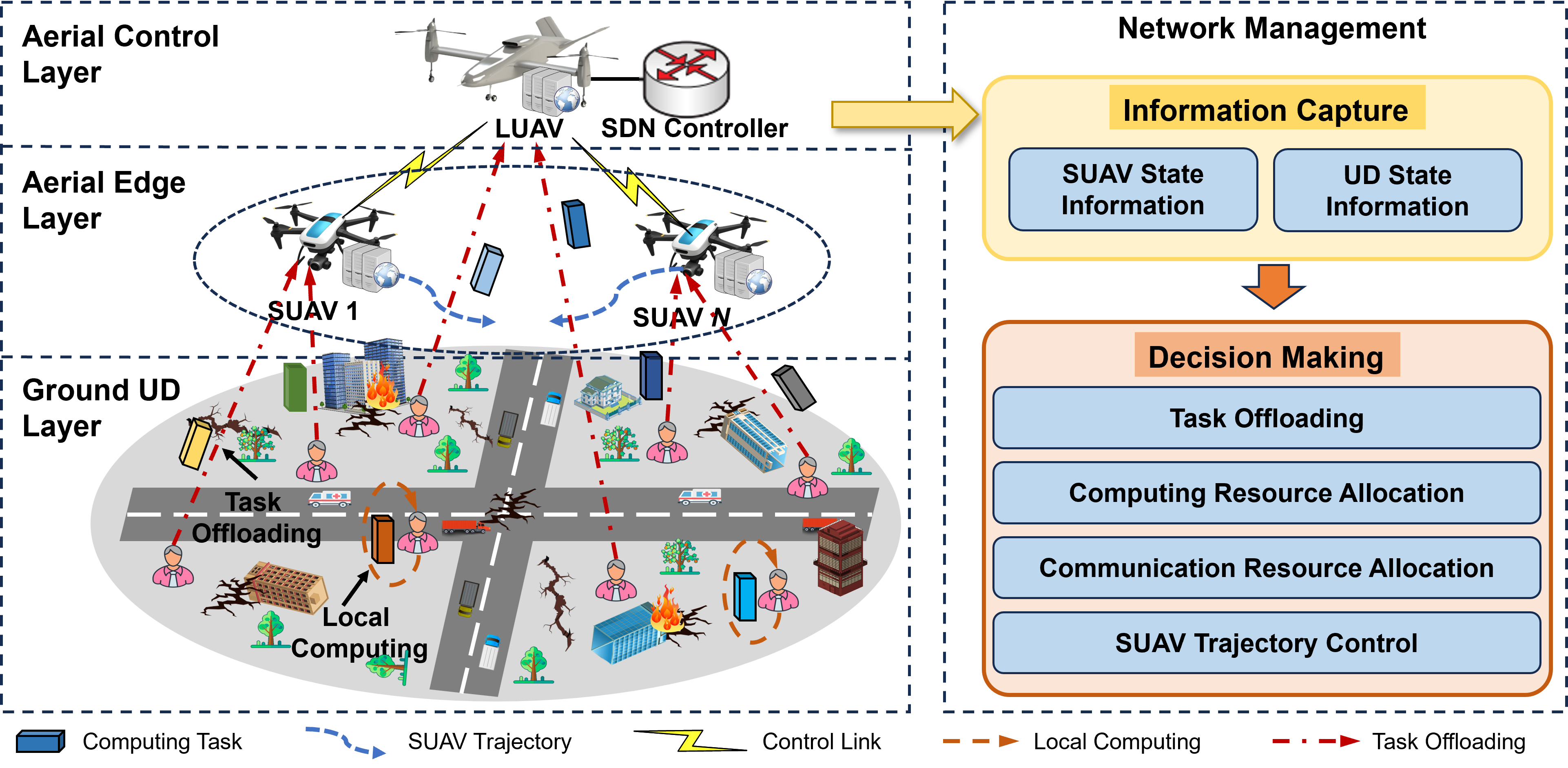}
	\caption{The hierarchical multiple-UAV-assisted MEC system consists of an aerial control layer, an aerial edge layer and a ground UD layer. Each ground UD can either execute computation tasks locally or offload them to an SUAV or the LUAV for processing. The SUAVs and the LUAV coordinate the allocation of communication and computation resources to receive and execute computation tasks from UDs. Moreover, the LUAV serves as the SDN controller, responsible for collecting system information and making corresponding decisions.}
	\label{fig_gameModel}
	\vspace{-1.5em}
\end{figure*}

\par As illustrated in Fig.~\ref{fig_gameModel}, we consider a hierarchical multiple-UAV-assisted MEC system, where multiple UAVs collaborate to provide aerial edge computing services to ground UDs in a disaster-stricken area. Specifically, in the spatial dimension, the hierarchical system comprises an aerial control layer, an aerial edge layer, and a ground UD layer.

\par \textit{At the aerial control layer}, a large rotary-wing UAV (LUAV) $u$ is deployed above the service area center to serve as a regional software-defined networking (SDN) controller. It performs the following essential functions: 1) offering wireless communication coverage for UAVs at the aerial edge layer and UDs at the ground UD layer; 2) provisioning reliable edge computing services for UDs; and 3) acting as a regional controller, on which our algorithm runs, to make real-time decisions based on the collected channel state information (CSI), as well as the state information of UDs and UAVs. \textit{At the aerial edge layer}, a set of small rotary-wing UAVs (SUAVs) $\mathcal{N}=\{1,\ldots,N\}$ is equipped with MEC computing capabilities to provide flexible edge computing services to UDs at the ground UD layer. \textit{At the ground UD layer}, a set of UDs $\mathcal{M}=\{1,\ldots,M\}$ moving within the considered area periodically generates computing tasks with heterogeneous and time-varying requirements.

\par In the temporal dimension, the system operates in a discrete time slot manner with $T$ equal time slots, i.e., $t\in \mathcal{T}=\{1,\ldots,T\}$, wherein each time slot has a duration of $\Delta t$. Here, $\Delta t$ is chosen to be sufficiently small such that each time slot can be considered as quasi-static. Furthermore, within each time slot, the state information of both UDs and SUAVs, as well as the CSI, are captured and updated. The task offloading, resource allocation and UAV trajectory control are determined by running our algorithm.

%
%
\subsection{Basic Models}
\label{subsec:Basic Models}

\par The basic models define the state information of various entities in the proposed system.

\par \textbf{\textit{UD Model.}} We consider that each UD generates one computing task per time slot~\cite{Wang2022}. For UD $m\in\mathcal{M}$, the  attributes of the UD at time slot $t$ can be characterized as $\mathbf{St}_m(t)=\left(f_m^{\text{UD}},\mathbf{\Phi}_m(t),\mathbf{q}_m(t)\right)$, where $f_m^{\text{UD}}$ denotes the local computing capability of UD $m$. The computing task generated by UD $m$ is characterized as $\mathbf{\Phi}_m(t)=\{D_m(t),\eta_m(t),T^{\text{max}}_m(t)\}$ at time slot $t$, wherein $D_m(t)$ represents the input data size (in bits), $\eta_m(t)$ denotes the computation intensity (in cycles/bit), and $T^{\text{max}}_m(t)$ is the maximum tolerable delay (in s). $\mathbf{q}_m(t)=\left[x_m(t),y_m(t)\right]$ represents the location coordinates of UD $m$ at time slot $t$. 

\par \textbf{\textit{UD Mobility Model.}} Similar to~\cite{Tabassum2019}, the mobility of UDs is modeled as a Gauss-Markov mobility model, which is widely employed in cellular communication networks. Specifically, the velocity of UD $m$ at time slot $t+1$ can be updated as
\begin{sequation}
    \mathbf{v}_m(t+1)=\alpha \mathbf{v}_m(t)+(1-\alpha)\overline{\mathbf{v}}_m+\sqrt{1-\alpha^2}\mathbf{w}_m(t),
    \label{eq.vehcle_velocity}
\end{sequation}

\noindent where $\mathbf{v}_m(t)=(v_m^x(t),v_m^y(t))$ denotes the velocity vector at time slot $t$. $\alpha$ represents the memory level, which reflects the temporal-dependent degree. $\overline{\mathbf{v}}_m$ is the asymptotic means of velocity. $\mathbf{w}_m(t)$ is the uncorrelated random Gaussian process $N(0,\sigma_m^2)$, where $\sigma_m$ denotes the asymptotic standard deviation of velocity. Therefore, the mobility of UD $m$ can be updated as follows:
\begin{sequation}
\mathbf{q}_m(t+1) = \mathbf{q}_m(t) + \mathbf{v}_m(t)\Delta t.
\end{sequation}

\par \textbf{\textit{SUAV Model.}} Each SUAV $n\in \mathcal{N}$ is characterized by $\mathbf{St}_{n}(t)=(\mathbf{q}_{n}(t),H,F_{n}^{\text{max}})$, wherein $\mathbf{q}_n(t) = [x_n(t),y_n(t)]$ and $H$ represent the horizontal coordinates and flight height of SUAV $n$ at time slot $t$, respectively. Moreover, $F_{n}^{\text{max}}$ denotes the total computing resources of SUAV $n$. 

\par \textbf{\textit{SUAV Mobility Model.}} Similar to~\cite{li2023robust}, we consider that each SUAV flies at a fixed altitude $H$ to avoid additional energy consumption caused by frequent ascent and descent. For each SUAV $n\in \mathcal{N}$, its trajectory can be expressed as a sequence of optimal positions for each time slot, i.e., $\mathcal{Q}_n=\{\mathbf{q}_n(t)\}_{t\in \mathcal{T}}$. In addition, the trajectory of each SUAV needs to satisfy the following several practical constraints:
\begin{subequations}
    \label{cons-traj}
    \begin{alignat}{3}
    &\mathbf{q}_n(1)=\mathbf{q}_n^{\text{ini}},\ \forall n\in \mathcal{N},\label{cons0}\\
    &\|\mathbf{q}_n(t+1)-\mathbf{q}_n(t)\|\leq v_n^{\text{max}}\Delta t,\ \forall n\in \mathcal{N},\ t\in\mathcal{T},\label{cons1}\\
    &\|\mathbf{q}_i(t)-\mathbf{q}_j(t)\|\geq d^{\text{min}},\ \forall i,j\in \mathcal{N},\ t\in\mathcal{T}, \label{cons2}
    \end{alignat}
\end{subequations}
\noindent where $\mathbf{q}_n^{\text{ini}}$ is the initial position of SUAV $n$, $v_n^{\text{max}}$ denotes the maximum flight speed of SUAV $n$, and $d^{\text{min}}$ denotes the minimum safe distance to avoid collision. Constraint~\eqref{cons0} defines the initial position of each SUAV. Constraint~\eqref{cons1} indicates that the flight distance of each time slot does not exceed the maximum allowable distance, and constraint~\eqref{cons2} ensures the safe flight of SUAVs to avoid collisions.

\par \textbf{\textit{LUAV Model.}} LUAV $u$ is characterized by $\mathbf{St}_{u}=(\mathbf{q}_{u},H_u,F_{u}^{\text{max}})$, wherein $\mathbf{q}_u = [x_u,y_u]$ and $H_u$ represent the horizontal coordinates and flight height, respectively. Moreover, $F_{u}^{\text{max}}$ denotes the total computing resources. 

%
%
\subsection{Communication Model}
\label{subsec:Communication Model}

\par The UDs can offload tasks to aerial servers\footnote{Note that an aerial server and a UAV will be used interchangeably.} via ground-to-air (G2A) links. To mitigate unreliable communication caused by interference, we assume that each aerial server utilizes different frequency band resources to provide edge computing services~\cite{10418937}. Furthermore, the widely used orthogonal frequency-division multiple access (OFDMA) technology can be employed to serve multiple UDs simultaneously~\cite{WeiCSNYZS21}. Therefore, there is no communication interference in the considered system. Assuming UD $m$ offloads tasks to server $s\in \mathcal{S}=\{u\}\cup \mathcal{N}$ at time slot $t$, the corresponding channel gain and transmission rate can be calculated as follows.

\subsubsection{Channel Gain Model} Considering that the G2A links may experience obstruction from environmental obstacles, the channel gain between UD $m$ and aerial server $s$ at time slot $t$ is calculated by incorporating the commonly used probabilistic LoS channel model~\cite{Zheng2024} with the large-scale and small-scale fadings as follows:
\begin{equation}
	g_{s,m}(t) = \rho_{s,m}^{\text{LoS}}(t) g_{s,m}^{\text{LoS}}(t) + (1-\rho_{s,m}^{\text{LoS}}(t)) g_{s,m}^{\text{NLoS}}(t),
\end{equation}

\noindent where $\rho_{s,m}^{\text{LoS}}(t)$ denotes the probability of LoS transmission. Moreover, $g_{s,m}^{x}(t)$ $(x\in \{\text{LoS}, \text{NLoS}\})$ denotes the channel gain under either LoS or non-LoS (NLoS) conditions, which is expressed as $g_{s,m}^x(t)=|h_{s,m}^x(t)|^2 L_{s,m}^x(t)^{-1}$, where $h_{s,m}^x(t)$ and $L_{s,m}^x(t)$ denotes the parameters of small-scale fading and large-scale fading, respectively~\cite{3GPP}.
\par Specifically, for the G2A links, the widely adopted LoS probability is calculated as~\cite{sun2023uav}
\begin{equation}
	\rho_{s,m}^{\text{LoS}}(t)=\frac{1}{1+c_1 \exp(-c_2(\frac{180}{\pi} \arcsin{\frac{H}{d_{s,m}(t)}} - c_1))},
\end{equation}

\noindent where $c_1$ and $c_2$ are the constants depending on the environment, and $d_{s,m}(t)$ means the straight-line distance between UD $m$ and server $s$. Moreover, the small-scale fading for G2A links at time slot $t$ is modeled as a parametric-scalable and good fitting generalized fading, i.e., Nakagami-$m$ fading~\cite{Khuwaja2018}, which is given as
\begin{equation}
	\begin{aligned}
		h_{s,m}^x(t) & \sim f^{\mathrm{Nak}}\left(h_{s,m}^{x}(t), \theta^x\right) \\
		& =\frac{2(\theta^x)^{\theta^x}}{\Gamma(\theta^x)(\bar{p})^{\theta^x}}(h_{s,m}^{x}(t))^{2 \theta^x-1} \exp\big(-\frac{\theta^x}{\bar{p}}(h_{s,m}^{x}(t))^2\big),
	\end{aligned}
\end{equation}

\noindent where $\bar{p}$ is the average received power, $\Gamma(\cdot)$ is the Gamma function, and $\theta^x$ $(x\in \{\text{LoS}, \text{NLoS}\})$ is the Nakagami-$m$ fading parameters for LoS/NLoS channel. Furthermore, the large-scale fading for the G2A links can be given as~\cite{Tian2023}
\begin{equation}
	L_{s,m}^x(t)=\frac{\left(4 \pi d_{s,m}(t) f_s\right)^2}{c^2}\eta^x,
\end{equation}

\noindent where $c$ denotes the speed of light, $f_s$ is the carrier frequency, and $\eta^x$ is the attenuation factor for LoS/NLoS channel.

\subsubsection{Transmission Rate Model} Based on the above analysis, the transmission rate between UD $m$ and server $s$ can be calculated using Shannon's formula as follows:
\begin{equation}
\label{eq.transratio1}
    R_{s,m}(t)=w_{s,m}B_{s}\log_2\left(1+\frac{p_mg_{s,m}(t)}{\varpi_0}\right),
\end{equation}

\noindent where $w_{s,m}\in (0,1]$ represents the resource allocation coefficient of UD $m$, $B_s$ denotes the bandwidth resources available to server $s$, $p_{m}$ is the transmission power of UD $m$, $g_{s,m}(t)$ represents the channel gain between UD $m$ and server $s$ at time slot $t$, and $\varpi_0$ is the noise power.

%
%
\subsection{Computation Model}
\label{subsec:Computation Model}

\par For task $\mathbf{\Phi}_m(t)$ generated by UD $m$, the task can be processed either locally on the UD or remotely on the aerial servers, which is determined by the offloading decision of the UD. To be more specific, we define a variable $a_m(t)\in \{0\}\cup\{u\}\cup \mathcal{N}$ to represent the task offloading decision of UD $m$ at time slot $t$, where $a_m(t)=0$ indicates that the task is executed locally, $a_m(t)=u$ represents that the task is offloaded to the LUAV for execution, and $a_m(t)=n\in\mathcal{N}$ means the task is offloaded to SUAV $n$ for execution.

%
%
\subsubsection{Local Computing Model}
\label{subsec:Local Computing Model}

\par If UD $m$ processes task $\mathbf{\Phi}_m(t)$ locally (i.e., $a_m(t)=0$), the task completion delay and UD energy consumption for local computing are given as follows.  

\par \textbf{Task Completion Delay.} The local completion delay of the task can be calculated as
\begin{sequation}
    \label{eq.loc-delay-T}
    T_m^{\text{loc}}(t)=\frac{\eta_m(t)D_m(t)}{f_m^{\text{UD}}},
\end{sequation}
\noindent where $f_m^{\text{UD}}$ represents the computing resources of UD $m$.

\par \textbf{UD Energy Consumption.} Accordingly, the energy consumption of UD $m$ to execute task $\mathbf{\Phi}_m(t)$ locally at time slot $t$ can be calculated as~\cite{UAV-H}
\begin{sequation}
    \label{eq.loc-delay-E}
    E_m^{\text{loc}}(t)=k(f_m^{\text{UD}})^{3}T_m^{\text{loc}}(t),
\end{sequation}

\noindent where $k$ represents the effective switched capacitance coefficient that depends on the hardware architecture~\cite{PanPWZW21}.

%
%

\subsubsection{Edge Computing Model}
\label{subsec:Edge Computing Model}

\par When task $\mathbf{\Phi}_m(t)$ is offloaded to an aerial server $s\in \mathcal{S}$ for processing (i.e., $a_m(t)=s$), the server allocates computing resources to perform the task and returns the processing results to the UD. Note that, similar to most existing studies~\cite{DengFW23,Sun2025a}, the latency of result feedback can be disregarded when considering the task completion delay. This is because for many intelligent applications (such as image recognition, real-time video analytics), the size of the processing results is typically significantly smaller than the size of the input data.

\par \textbf{Task Completion Delay.} For edge computing, the task completion delay mainly includes task transmission delay and edge execution delay, which can be calculated as
\begin{sequation}
    \label{eq.ec-delay}
    T_{s,m}^{\text{ec}}(t) = \frac{D_m(t)}{R_{s,m}t)}+\frac{\eta_m(t) D_m(t)}{F_{s,m}(t)},
\end{sequation}

\noindent where $F_{s,m}(t)$ denotes the computing resources allocated by server $s$ to UD $m$ at time slot $t$.

\par \textbf{UD Energy Consumption.} The task offloading of the UD incurs transmission energy consumption at time slot $t$, which can be calculated as
\begin{sequation}
    \label{eq.trans-energy}
    E_{s,m}^{\text{ec}}(t) = p_m\frac{D_m(t)}{R_{s,m}(t)},
\end{sequation}

\noindent where $p_m$ represents the transmission power of UD $m$.

\par \textbf{SUAV Energy Consumption.} If task $\mathbf{\Phi}_m(t)$ is offloaded to SUAV $n$, the SUAV incurs computation energy consumption to execute the task, which can be given as~\cite{JiangDXI23}
\begin{sequation}
    \label{eq.uav-comp-energy}
    E_{n,m}^{\text{c}}(t) = \varpi\eta_m(t)D_m(t),
\end{sequation}

\noindent where $\varpi$ represents the energy consumption per unit CPU cycle of SUAV $n$. Therefore, the total computational energy consumption of SUAV $n$ at time slot $t$ can be given as
\begin{sequation}
    \label{eq.comp-energy}
    E_n^{\text{c}}(t) = \sum_{m\in \mathcal{M}}I_{\{a_m(t)=n\}}E_{n,m}^{\text{c}}(t),
\end{sequation}

\noindent where $I_{\{X\}}$ is an indicator function defined such that $I_{\{X\}}=1$ when event $X$ is true, and $I_{\{X\}}=0$ otherwise. Furthermore, the trajectory control of SUAV $n$ incurs corresponding propulsion energy consumption at time slot $t$. Similar to~\cite{ pan2023joint}, the propulsion power consumption of SUAV $n$ at time slot $t$ can be expressed as
\begin{sequation}
\label{eq.prop-energy}
    \begin{aligned}
    P_n(v_n(t))=&\underbrace{C_1\left(1+\frac{3 v_n(t)^2}{U_{\text{p}}^2}\right)}_{\text {blade profile }}+\underbrace{C_4 v_n(t)^3}_{\text {parasite}}\\
    &+\underbrace{C_2 \sqrt{\sqrt{C_3+\frac{v_n(t)^4}{4}}-\frac{v_n(t)^2}{2}}}_{\text{induced }},
    \end{aligned}
\end{sequation}

\noindent where $v_n(t)$ represents the velocity of SUAV $n$ at time slot $t$, $U_{\text{p}}$ refers to the rotor's tip speed, and $C1$, $C2$, $C3$, and $C4$ are the constants described in~\cite{Yang2022}. Therefore, the energy consumption of SUAV $n$ at time slot $t$ can be given as
\begin{sequation}
    \label{eq.UAV-energy}
    E_n(t) = E_n^{\text{c}}(t) + E_n^{\text{p}}(t),
\end{sequation}

\noindent where $E_n^{\text{p}}(t)=P_n(v_n(t))\tau$ denotes the propulsion energy consumption at time slot $t$. To guarantee service time, we define the SUAV energy consumption constraint as follows:
\begin{sequation}
    \label{eq.eng-cons}
    \lim _{T \rightarrow+\infty} \frac{1}{T} \sum_{t=1}^{T} E_n(t) \leq \bar{E}_n,\ \forall n \in \mathcal{N},
\end{sequation}

\noindent where $\bar{E}_n$ is the energy budget of the SUAV $n$ per time slot.

\par Note that if task $\mathbf{\Phi}_m(t)$ is offloaded to server $u$, we ignore the energy consumption of server $u$. This is because server $u$ is a large UAV with a more sufficient energy supply, which can provide a longer service time compared to server $n\in \mathcal{N}$.

%
%
\subsection{QoE Evaluation Model}
\label{subsec:QoE Evaluation Model}

\par Given the limited battery capacity of UDs and the latency sensitivity of computing tasks, task completion delay and UD energy consumption are important performance indicators for evaluating the QoE of UDs. Similar to~\cite{Chen2022}, we combine task completion latency and UD energy consumption into a cost function to measure the QoE of UDs. Specifically, the completion delay of task $\mathbf{\Phi}_m(t)$ at time slot $t$ can be given as
\begin{sequation}
    \label{eq.delay}
    T_m(t)= I_{\{a_m(t)=0\}}T_m^{\text{loc}}(t)+\sum_{s\in \mathcal{S}}I_{\{a_m(t)=s\}}T_{m,s}^{\text{ec}}(t).
\end{sequation}

\noindent Then, the energy consumption of UD $m$ at time slot $t$ can be given as
\begin{sequation}
    \label{eq.energy}
    E_m(t)= I_{\{a_m(t)=0\}}E_m^{\text{loc}}(t)+\sum_{s\in \mathcal{S}}I_{\{a_m(t)=s\}}E_{m,s}^{\text{ec}}(t).
\end{sequation}

\noindent According to~\cite{Chen2022}, the cost of UD $m$ at time slot $t$ can be formulated as
\begin{sequation}
    \label{eq.UD-cost}
    C_m(t)=\gamma_m^{\text{T}}T_m(t) + \gamma_m^{\text{E}}E_m(t),
\end{sequation}

\noindent where $\gamma_m^{\text{T}}$ and $\gamma_m^{\text{E}}$ respectively represent the weight coefficients of task completion delay and energy consumption for UD $m$, which can be flexibly set based on the preference of the UD for delay and energy consumption. Clearly, minimizing the cost of UDs is equivalent to maximizing the QoE of UDs.

%
%
\subsection{Problem Formulation}
\label{subsec:problem Formulation}

\par The objective of this work is to minimize the average cost of all UDs over time (i.e., the time-averaged UD cost), by jointly optimizing the task offloading $\mathbf{A}=\{\mathcal{A}^t|\mathcal{A}^t=\{a_m(t)\}_{m\in\mathcal{M}}\}_{t\in\mathcal{T}}$, computing resource allocation $\mathbf{F}=\{\mathcal{F}^t|\mathcal{F}^t=\{F_{s,m}(t)\}_{s\in\mathcal{S},m\in\mathcal{M}}\}_{t\in\mathcal{T}}$, communication resource allocation $\mathbf{W}=\{\mathcal{W}^t|\mathcal{W}^t=\{w_{s,m}(t)\}_{s\in \mathcal{S},m\in\mathcal{M}}\}_{t\in \mathcal{T}}$, and trajectory control $\mathbf{Q}=\{\mathcal{Q}^t|\mathcal{Q}^t=\{\mathbf{q}_n(t)\}_{n\in\mathcal{N}}\}_{t\in \mathcal{T}}$. Therefore, the problem can be formulated as follows:

\vspace{-0.8em}
{\small
 \begin{align}
    \textbf{P}: \quad &\underset{\mathbf{A}, \mathbf{F}, \mathbf{W}, \mathbf{Q}}{\text{min}} \frac{1}{T}\sum_{t=1}^{T}\sum_{m=1}^{M}C_m(t) \label{P}\\
    \text{s.t.}\ \ 
    &\lim _{T \rightarrow+\infty} \frac{1}{T} \sum_{t=1}^{T} E_n(t) \leq \bar{E}_n, \forall n\in \mathcal{N}, \tag{\ref{P}{\text{a}}} \label{Pa}\\
    &a_m(t)\in \{0\}\cup\mathcal{S}, \forall m\in \mathcal{M}, t\in \mathcal{T}, \tag{\ref{P}{\text{b}}} \label{Pb}\\
    &I_{\{a_m(t)=s\}}T^{\text{ec}}_{s,m}(t)\leq T_m^{\text{max}}, \forall m\in \mathcal{M}, s\in \mathcal{S}, t\in \mathcal{T}, \tag{\ref{P}{\text{c}}} \label{Pc}\\
    &0\leq F_{s,m}(t) \leq F_s^{\text{max}}, \forall m\in \mathcal{M}, s\in \mathcal{S}, t\in \mathcal{T}, \tag{\ref{P}{\text{d}}} \label{Pd}\\
    &\sum_{m=1}^M I_{\{a_m(t)=s\}}F_{s,m}(t)\leq F_s^{\text{max}}, \forall s\in \mathcal{S}, t\in \mathcal{T}, \tag{\ref{P}{\text{e}}} \label{Pe}\\
    &0\leq w_{s,m}(t) \leq 1, \forall m\in \mathcal{M}, s\in \mathcal{S}, t\in \mathcal{T}, \tag{\ref{P}{\text{f}}} \label{Pf}\\
    &\sum_{m=1}^M I_{\{a_m(t)=s\}} w_{s,m}(t)\leq 1, \forall s\in \mathcal{S}, t\in \mathcal{T}, \tag{\ref{P}{\text{g}}} \label{Pg}\\
    &(\rm \ref{cons0})-(\rm \ref{cons2}). \tag{\ref{P}{\text{h}}} \label{Ph}
\end{align}}

\noindent Constraint (\ref{Pa}) ensures that the long-term energy consumption of each SUAV does not exceed its designated energy budget. Constraint (\ref{Pb}) defines the strategy space for task offloading decisions of each UD. Constraint (\ref{Pc}) guarantees that the task completion delay for edge computing does not exceed the maximum tolerable latency of the task. Moreover, constraints (\ref{Pd}) and (\ref{Pe}) specify that the amount of computing resources allocated by each aerial server must be positive and should not exceed its total available computing resources. In addition, constraints (\ref{Pf}) and (\ref{Pg}) require that the communication resource allocation coefficients of each aerial server be positive and remain within the available communication resources of the server. Finally, constraint (\ref{Ph}) limits the initial positions of SUAVs, the maximum flying distance per time slot, and the minimum safe distance between any pair of SUAVs.

\par \textbf{\textit{Challenges.}} There are two main challenges to obtain the optimal solution of problem $\textbf{P}$. \textit{i) Future-dependent.} Optimally solving problem $\textbf{P}$ requires complete future information, e.g., task computing demands and locations of all UDs across all time slots. However, obtaining the future information is very challenging in the considered time-varying scenario. \textit{ii) Non-convex and NP-hard.} Problem $\textbf{P}$ contains both continuous variables (i.e., resource allocation $\{\mathbf{F},\mathbf{W}\}$ and trajectory control $\mathbf{Q}$) and binary variables (i.e., task offloading decision $\mathbf{A}$) is a mixed-integer non-linear programming (MINLP) problem, which is non-convex and NP-hard~\cite{boyd2004convex,belotti2013mixed}. To this end, we proposed OJTRTA to solve this problem, which consists of a Lyapunov-based problem transformation process and a two-stage optimization procedure. Fig.~\ref{fig_algrithm_framework} shows the framework of the proposed OJTRTA.
\begin{figure*}[!hbt]
	\centering
	\setlength{\abovecaptionskip}{0pt}%
	\setlength{\belowcaptionskip}{0pt}%
	\includegraphics[width =6.8in]{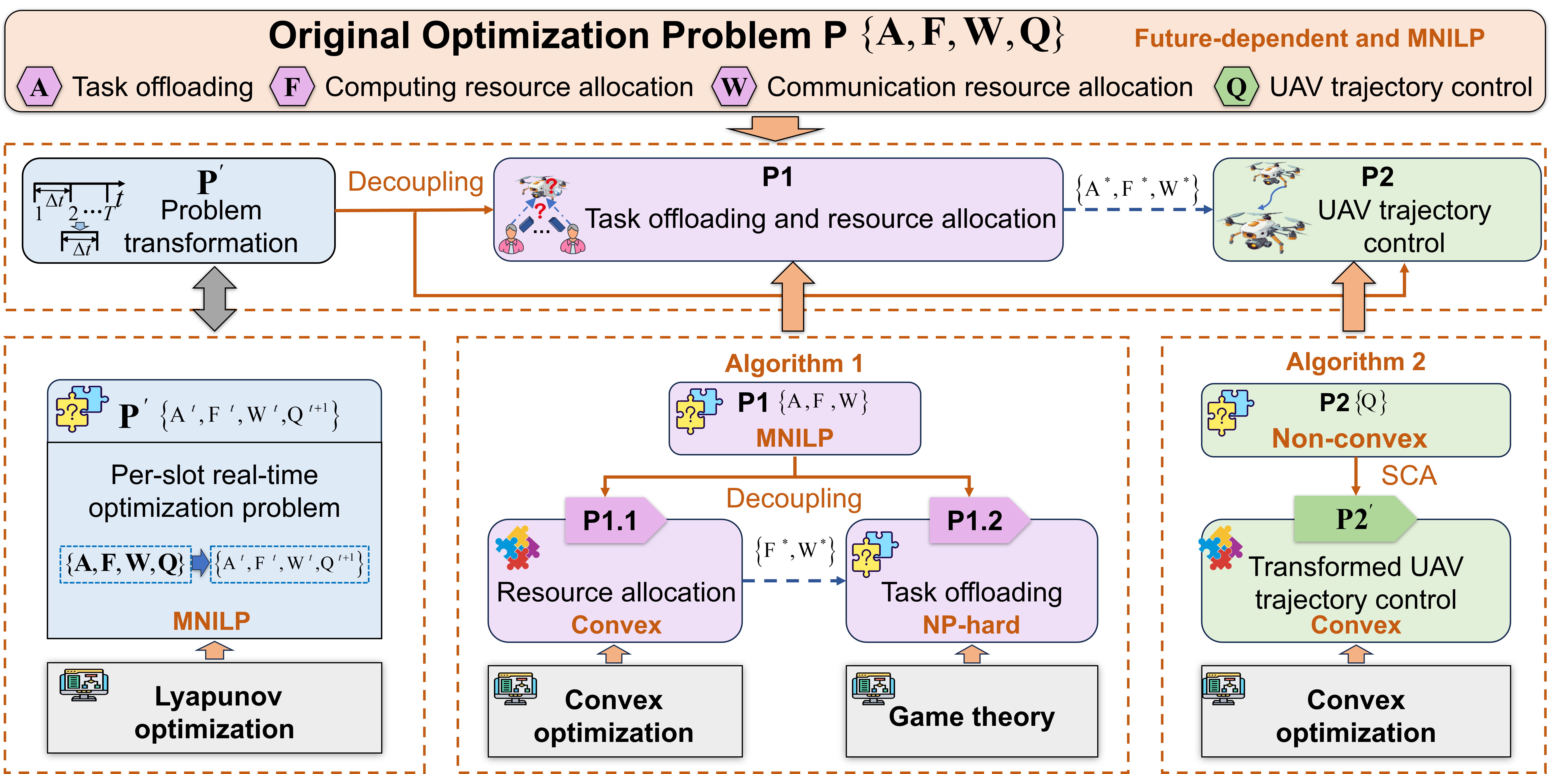}
	\caption{The framework of OJTRTA. The original problem $\textbf{P}$, which depends on future information, is first transformed into a per-slot real-time optimization problem $\textbf{P}^{\prime}$ using the Lyapunov optimization framework. Then, the problem $\textbf{P}^{\prime}$ is decomposed into two subproblems: the task offloading and resource allocation subproblem $\textbf{P1}$, and the UAV trajectory control subproblem $\textbf{P2}$. For $\textbf{P1}$, the optimal resource allocation strategy is obtained via convex optimization, while the task offloading decision is obtained using game-theoretic methods. For $\textbf{P2}$, the UAV trajectory control strategy is obtained by reformulating the problem into a convex optimization form via the SCA technique, which is subsequently solved by convex optimization.} 
	\label{fig_algrithm_framework}
	\vspace{-1.5em}
\end{figure*}

%
%
\section{Lyapunov-Based Problem Transformation}
\label{sec:Lyapunov-Based problem Transformation}

\par In this section, we first explain the motivation for designing an online approach using the Lyapunov optimization framework. Then, we provide a detailed explanation of the Lyapunov-based problem transformation.

%
%
\subsection{Motivation}
\label{subsec:Motivation}

\par Since problem $\textbf{P}$ is future-dependent, an online approach is necessary to make real-time decisions without foreseeing the future. The Lyapunov-based optimization framework and DRL are two practical tools for designing online approaches. However, DRL typically requires a substantial number of samples to acquire optimal policies and faces challenges in achieving convergence. Furthermore, DRL requires significant computational resources, making it costly in the resource-constrained MEC system. Compared to DRL, the Lyapunov-based optimization framework offers the following advantages:
\begin{itemize}
    \item \textbf{\textit{Simplicity and Interpretability:}} The Lyapunov-based optimization framework usually has a simpler structure and are easier to interpret compared to DRL. Furthermore, the Lyapunov-based optimization framework does not introduce additional complexity to the algorithm design.
    \item \textbf{\textit{Stability:}} The Lyapunov-based optimization framework provides stability guarantees, which ensures convergence and boundedness of system dynamics. 
    \item \textbf{\textit{Model-free Nature:}} The Lyapunov-based optimization framework does not require explicit knowledge of the system dynamics, making it suitable for the considered scenario where obtaining future information about system dynamics is challenging and impractical.
\end{itemize}

\par Therefore, we adopt the Lyapunov-based optimization framework for online approach design. Furthermore, in the simulation experiments, we compare the proposed Lyapunov-based online approach with DRL-based approaches to verify its suitability for the considered scenarios.

%
%
\subsection{Problem Transformation}
\label{subsec:Problem Transformation}
\par Firstly, to satisfy the SUAV energy consumption constraint (\ref{Pa}), we define two virtual energy queues $Q_n^{\text{c}}(t)$ and $Q_n^{\text{p}}(t)$ to represent the computing energy queue and the propulsion energy queue at time slot $t$ based on Lyapunov optimization technique, respectively. We assume that the queues are set as zero at the initial time slot, i.e., $Q_n^{\text{c}}(1)=0$ and $Q_n^{\text{p}}(1)=0$. Therefore, the virtual energy queues can be updated as
\begin{sequation}
    \label{eq.queue}
    \begin{cases}
        Q_n^{\text{c}}(t+1)=\max \left\{Q_n^{\text{c}}(t)+E_n^{\text{c}}(t)-\bar{E_n^{\text{c}}}, 0\right\}, \forall n \in \mathcal{N}, t \in \mathcal{T},\\
        Q_n^{\text{p}}(t+1)=\max \left\{Q_n^{\text{p}}(t)+E_n^{\text{p}}(t)-\bar{E_n^{\text{p}}}, 0\right\}, \forall n \in \mathcal{N}, t \in \mathcal{T},
    \end{cases}
\end{sequation}

\noindent where $\bar{E_n^{\text{c}}}$ and $\bar{E_n^{\text{p}}}$ denote the computation and propulsion energy budgets per slot, respectively, and $\bar{E_n^{\text{c}}}+\bar{E_n^{\text{p}}}=\bar{E_n}$. Secondly, we define the \textit{Lyapunov function} $L(\boldsymbol{\Theta}(t))$, which represents a scalar measure of the queue backlogs, i.e.,
\begin{sequation}
    \label{eq.ly-func}
    L(\boldsymbol{\Theta}(t)) = \frac{1}{2}\sum_{n=1}^{N}(Q_n^{\text{c}}(t))^2+\frac{1}{2}\sum_{n=1}^{N}(Q_n^{\text{p}}(t))^2,
\end{sequation}

\noindent where $\boldsymbol{\Theta}(t)=[\boldsymbol{Q^{\text{c}}(t)},\boldsymbol{Q^{\text{p}}(t)}]$ is the vector of current queue backlogs. Thirdly, we define the \textit{conditional Lyapunov drift} for time slot $t$ as
\begin{sequation}
    \label{eq.cond-ly-drift}
    \Delta L(\boldsymbol{\Theta}(t)) \triangleq \mathbb{E}\{L(\boldsymbol{\Theta}(t+1))-L(\boldsymbol{\Theta}(t)) \mid \boldsymbol{\Theta}(t)\}.
\end{sequation}

\noindent Finally, similar to~\cite{2010Neely,JiangDXI23,Yang2022}, the \textit{drift-plus-penalty} can be given as
\begin{sequation}
    \label{eq.drift-plus-penalty}
   D(\boldsymbol{\Theta}(t)) =\Delta L(\boldsymbol{\Theta}(t))+V \mathbb{E}\left\{C_s(t)\mid \mathbf{\Theta}(t)\right\},
\end{sequation}

\noindent where $C_s(t)=\sum_{m=1}^{M}C_m(t)$ is the total cost of all UDs at time slot $t$, and $V$ is a parameter to trade off the total cost and the queue stability. Next, we provide an upper bound on the \textit{drift-plus-penalty}, as stated in Theorem~\ref{the:drift-plus-penalty}.

%
%
\begin{theorem}
\label{the:drift-plus-penalty}
For all $t$ and all possible queue backlogs $\boldsymbol{\Theta}(t)$, the drift-plus-penalty is upper bounded as
\begin{sequation}
    \label{eq.theorem1}
    \begin{aligned}
    D(\mathbf{\Theta}(t)) \leq & W + \sum_{n=1}^{N}{Q}_n^{\mathrm{c}}(t)(E_n^{\mathrm{c}}(t)-\bar{E_n^{\mathrm{c}}})\\
    &+\sum_{n=1}^{N}{Q}_n^{\mathrm{p}}(t)(E_n^{\mathrm{p}}(t)-\bar{E_n^{\mathrm{p}}})+V\cdot C_s(t),
    \end{aligned}
\end{sequation}

\noindent where $W=\frac{1}{2} \sum_{n=1}^{N}\max \left\{\left(\bar{E_n^{\mathrm{c}}}\right)^2,\left(E_{\max }^{\mathrm{c}}-\bar{E_n^{\mathrm{c}}}\right)^2\right\}+\frac{1}{2} \sum_{n=1}^{N}\max \left\{\left(\bar{E_n^{\mathrm{p}}}\right)^2,\left(E_{\max }^{\mathrm{p}}-\bar{E_n^{\mathrm{p}}}\right)^2\right\}$ is a finite constant.
\end{theorem}

%
%
\begin{proof}
Refer to Section I of the supplementary material.
\end{proof}

\par According to the Lyapunov optimization framework, we minimize the right-hand side of inequality (\ref{eq.theorem1}). Therefore, problem $\textbf{P}$ relying on future information is transformed into the real-time optimization problem $\textbf{P}^{'}$ solvable with only current information, which is given as follows:

\vspace{-0.8em}
{\small \begin{align}
\textbf{P}^{\prime}:&\underset{\mathcal{A}^t,\mathcal{F}^t,\mathcal{W}^t,\mathcal{Q}^{t+1}}{\text{min}}\sum_{n=1}^{N}\left(Q_n^{\text{c}}(t)E_n^{\text{c}}(t)+Q_n^{\text{p}}(t)E_n^{\text{p}}(t)\right)+V\sum_{m=1}^{M} C_m(t) \label{P_temp} \\
\text{s.t.} \ &(\ref{Pb})-(\ref{Ph}). \notag
\end{align}}

\noindent However, problem $\textbf{P}^{\prime}$ is still a MINLP problem, with interdependencies among the decision variables. Note that for this complex problem $\textbf{P}^{\prime}$, obtaining the optimal solution may not be suitable for the considered scenario. First, although the optimal solution can ensure the maximization QoE for all UDs, it may result in unfairness among UDs. Second, obtaining the optimal solution may result in significant computational overhead, which is difficult to meet the needs of real-time decision making. To this end, we propose a two-stage optimization method that uses game theory to determine the offloading decisions of UDs to ensure fairness among UDs. Moreover, the proposed method can also achieve a high-performance solution within a low computational complexity. The proposed two-stage optimization method is described in detail in the following section. Furthermore, for the sake of convenience in subsequent discussions, we adopt a notation where the time index is omitted.

%
%
\section{Two-Stage Optimization method}
\label{sec:Two-Stage Optimization Algorithm}
\par In this section, a two-stage optimization method is proposed to solve the transformed problem $\mathbf{P}^\prime$. In the first stage, assuming a feasible $\mathcal{Q}$, we optimize the task offloading $\mathcal{A}$ and resource allocation $\{\mathcal{F},\mathcal{W}\}$. In the second stage, based on the obtained task offloading $\mathcal{A}^{*}$ and resource allocation $\{\mathcal{F}^{*},\mathcal{W}^{*}\}$, we optimize the positions $\mathcal{Q}^{*}$ of SUAVs. It should be noted that although we optimize the resource allocation, task offloading, and UAV trajectory control in a decoupled manner, the decision variables related to these aspects are still considered together throughout the decoupling process. This is because optimizing one decision variable is conducted under the assumption of given values for the other two variables.
%
%
\subsection{Stage 1: Task Offloading and Resource Allocation}
\label{subsec:Task Offloading Decision}

\par Assuming a feasible $\mathcal{Q}$ and removing irrelevant constant terms, $\textbf{P}^{'}$ is transformed into the subproblem $\textbf{P1}$ to decide task offloading and resource allocation, which is given as follows:

 {\small \begin{align}
    \textbf{P1}:\ &V\cdot \underset{\mathcal{A},\mathcal{F},\mathcal{W}}{\text{min}}\left(\sum_{n=1}^{N}\frac{Q_n^{\text{c}}}{V}E_n^{\text{c}}+\sum_{m=1}^{M}C_m\right) \label{P1}\\
    \text{s.t.}\ &(\rm \ref{Pb})-(\rm \ref{Pg}). \notag
\end{align}}

\noindent Problem $\textbf{P1}$ is still a MINLP problem, and task offloading and resource allocation are coupled with each other. Considering that the UAVs are service providers in the considered UAV-assisted MEC system, we prioritize resource allocation strategies for the UAVs. Then, based on the resource allocation strategies, we optimize the offloading decisions of UDs.

%
%
\subsubsection{Resource Allocation}
\label{subsubsec:Resource Allocation}

\par For any arbitrary task offloading decision profile $\mathcal{A}$ of the UDs, the servers decides resource allocation strategies to minimize problem $\textbf{P1}$. Define $z_{s,m}=\frac{F_{s,m}}{F_s^{\text{max}}}$ and remove irrelevant constant terms, the resource allocation subproblem can be formulated as

\vspace{-0.8em}
{\small
\begin{align}
    \textbf{P1.1}:\ &\underset{\mathcal{Z},\mathcal{W}}{\text{min}}\sum_{s\in\mathcal{S}}\sum_{m\in\mathbf{M}_s(\mathcal{A})}\left(\frac{\gamma_m^{\text{T}}\eta_mD_m}{z_{s,m}F_s^{\text{max}}}+\frac{\gamma_m^{\text{T}}D_m+\gamma_m^{\text{E}}p_mD_m}{w_{s,m}r_{s,m}}\right) \label{P1.1}\\
    \text{s.t.}\ \
    & z_{s,m}\geq 0, \forall s\in \mathcal{S},\ m\in \mathbf{M}_s(\mathcal{A}), \tag{\ref{P1.1}{a}} \label{P1.1a}\\
    & \sum_{m\in\mathbf{M}_s(\mathcal{A})} z_{s,m}\leq 1, \forall s\in \mathcal{S}, \tag{\ref{P1.1}{b}} \label{P1.1b}\\
    & w_{s,m}\geq 0, \forall s\in \mathcal{S},\ m\in \mathbf{M}_s(\mathcal{A}), \tag{\ref{P1.1}{c}} \label{P1.1c}\\
    & \sum_{m\in\mathbf{M}_s(\mathcal{A})} w_{s,m}\leq 1, \forall s\in \mathcal{S}, \tag{\ref{P1.1}{d}} \label{P1.1d}
\end{align}}

\noindent where $r_{s,m}=B_{s}\log_2\left(1+\frac{p_mg_{s,m}(t)}{\varpi_0}\right)$, $\mathcal{Z}=\{z_{s,m}\}_{s\in \mathcal{S}, m\in\mathbf{M}_s(\mathcal{A})}$, and $\mathbf{M}_s(\mathcal{A})$ represents the set of UDs that offload tasks to server $s$, which is determined by the offloading decisions $\mathcal{A}$.

\par To solve problem $\textbf{P1.1}$, we first prove that the problem is a convex optimization problem through Lemma~\ref{lem-convex}. Then, we can obtain the optimal resource allocation by Theorem~\ref{them:resource-allocation}.

%
%
\begin{lemma}
\label{lem-convex}
Problem $\textbf{P1.1}$ is a convex optimization problem. 
\end{lemma}

%
%
\begin{proof}
Refer to Section II of the supplementary material.
\end{proof}

%
%
\begin{theorem}
\label{them:resource-allocation}
The optimal resource allocation coefficient, i.e., the solution of problem $\textbf{P1.1}$, can be given as follows:
\begin{sequation}
\label{eq.solution}
    \begin{cases}
        z_{s,m}^{*} = \frac{\sqrt{\frac{\gamma_m^{\mathrm{T}}\eta_mD_m}{F_s^{\mathrm{max}}}}}{\sum_{i\in \mathbf{M}_s(\mathcal{A})}\sqrt{\frac{\gamma_i^{\mathrm{T}}\eta_iD_i}{F_s^{\mathrm{max}}}}},\\
        w_{s,m}^{*} = \frac{\sqrt{\frac{\gamma_m^{\mathrm{T}}D_m+\gamma_m^{\mathrm{E}}p_mD_m}{r_{s,m}}}}{\sum_{i\in \mathbf{M}_s(\mathcal{A})}\sqrt{\frac{\gamma_i^{\mathrm{T}}D_i+\gamma_i^{\mathrm{E}}p_iD_i}{r_{s,i}}}}.
    \end{cases}
\end{sequation}
\end{theorem}

%
%
\begin{proof}
Refer to Section III of the supplementary material.
\end{proof}

%
%
\subsubsection{Task Offloading}
\label{subsubsec:Task Offloading}

\par For UD $m$, let us define $U_m^{\text{loc}}$ as the utility of local computing, $U_m^{\text{SUAV}}$ as the utility of SUAV-assisted computing, and $U_m^{\text{LUAV}}$ as the utility of LUAV-assisted computing, which can be given as follows:

\vspace{-0.8em}
{\small\begin{align}
&U_m^{\text{loc}}=\gamma_m^{\text{T}}T_m^{\text{loc}}+\gamma_m^{\text{E}}E_m^{\text{loc}},\label{eq.u_loc}\\
&U_m^{\text{SUAV}}(\mathcal{A})=\frac{Q_n^{\text{c}}}{V}E_{n,m}^{\text{c}}+\gamma_m^{\text{T}}T_m^{\text{SUAV}}+\gamma_m^{\text{E}}E_m^{\text{SUAV}}.\label{eq.u_sec}\\
&U_m^{\text{LUAV}}(\mathcal{A})=\gamma_m^{\text{T}}T_m^{\text{LUAV}}+\gamma_m^{\text{E}}E_m^{\text{LUAV}}.\label{eq.u_lec}
\end{align}}

\noindent Therefore, we can design the utility function of UD $m$ as
\begin{equation}
   \label{eq.utility}
   U_m(\mathcal{A}) =
       \begin{cases}
           U_m^{\text{loc}},a_m=0,\\
           U_m^{\text{SUAV}}(\mathcal{A}),a_m = n,\forall n\in\mathcal{N},\\
           U_m^{\text{LUAV}}(\mathcal{A}),a_m=u.
       \end{cases} \\
\end{equation}

\noindent According to the optimal resource allocation policy $\{\mathcal{F}^{*},\mathcal{W}^{*}\}$ and removing irrelevant constant terms, problem $ \textbf{P1}$ can be reformulated into a task offloading subproblem $\textbf{P1.2}$ as follows:

\vspace{-0.8em}
 {\small \begin{align}
    \textbf{P1.2}:\quad &\underset{\mathcal{A}}{\text{min}}\sum_{m\in \mathcal{M}}U_m(\mathcal{A}) \label{P1.2}\\
    \text{s.t.}\ \
    &(\ref{Pb})\ \text{and}\ (\ref{Pc}). \notag
\end{align}}

\par However, as shown in Eq.~\eqref{eq.solution}, the task offloading decision for UD $m$ is influenced by not only its individual demand, but also the offloading decisions made by other UDs. Considering the inherent competitive dynamics of task offloading among UDs, we utilize game theory to address this subproblem.

\par \textbf{(1) Game Formulation.} We first model the task offloading subproblem as a multi-UDs task offloading game (MUD-TOG). Specifically, the MUD-TOG can be defined as a triplet $\Gamma=\{\mathcal{M},\mathbb{A}, (U_m(\mathcal{A}))_{m\in \mathcal{M}}\}$, which is detailed as follows:

\begin{itemize}
\item $\mathcal{M}=\{1,2,\dots,M\}$ denotes the set of players, i.e., all UDs.
\item $\mathbb{A}=\mathbf{A}_1\times\dots\times\mathbf{A}_M$ denotes the strategy space, wherein $\mathbf{A}_m=\{0,u\}\cup{\mathcal{N}}$ is the set of offloading strategies for player $m\ (m\in \mathcal{M})$, $a_m\in\mathbf{A}_m$ denotes the offloading decision of player $m$, and $\mathcal{A}=(a_1,\dots,a_M)\in \mathbb{A}$ denotes a strategy profile.
\item $(U_m(\mathcal{A}))_{m\in \mathcal{M}}$ is the utility function of player $m$ that assigns a real number to each strategy profile $\mathcal{A}$.
\end{itemize}

\noindent Every player strives to minimize its utility by selecting an appropriate strategy. Mathematically, the MUD-TOG can be described as the following distributed optimization problem
\begin{sequation}
    \label{eq.task-offloading}
    \underset{a_m}{\text{min}}\ U_m(a_m,a_{-m}),\ \forall m \in \mathcal{M},
\end{sequation}

\noindent where $a_{-m}=(a_1,\dots,a_{m-1},a_{m+1},\dots,a_M)$ denotes the offloading decisions of the other players except player $m$.

\par \textbf{(2) The Solution of MUD-TOG.} We begin by introducing the concept of Nash equilibrium, which is commonly used to describe the solutions of MUD-TOG. A Nash equilibrium signifies a strategic profile wherein no individual player possesses the motivation to unilaterally change their current strategy, which can be formally defined as Definition~\ref{def:def1}.
\begin{definition}
\label{def:def1}
If and only if a strategy profile $\mathcal{A}^*=(a_1^*,\dots,a_M^*)$ satisfies the following condition, it can be considered as a Nash equilibrium of game $\Gamma$
    \begin{sequation}
    U_m(a_m^*,a_{-m}^*)\leq U_m(a_m^{\prime},a_{-m}^*) \quad  \forall a_m^{\prime}\in \mathbf{A}_m, m \in \mathcal{M}.
    \end{sequation}
\end{definition}

\par Next, we introduce an important framework called the exact potential game~\cite{potential} through Definitions~\ref{def:def2},~\ref{def:def_FIP}, and~\ref{def:FIP}, to analyze whether there is a Nash equilibrium for MUD-TOG and how to obtain a Nash equilibrium.
\begin{definition}
    \label{def:def2}
    If the game $\Gamma$ has a potential function $F(\mathcal{A}): \mathbb{A} \mapsto \mathbb{R}$ that satisfies the following condition, it can be regarded as an exact potential game.  
    \begin{sequation}
    \label{PG-def}
    \begin{split}
        &U_m(a_m,a_{-m})-U_m(a_m^{\prime},a_{-m})=F(a_m,a_{-m})-F(a_m^{\prime},a_{-m}), \\ &\forall (a_m,a_{-m}),(a_m^{\prime},a_{-m})\in \mathbb{A}.   
    \end{split}
    \end{sequation}
\end{definition}

\begin{definition}
	A Nash equilibrium and the finite improvement path (FIP) property always exist for an exact potential game with finite strategy sets
~\cite{potential,2016Potential}. 
    \label{def:def_FIP}
\end{definition}

\begin{definition}
	\label{def:FIP}
	A sequence of strategy profiles $\rho = (\mathcal{A}^0, \mathcal{A}^1, \mathcal{A}^2, \ldots)$ is called an improvement path if for every index $i\geq 0$, $\mathcal{A}^{i+1}$ is obtained from $\mathcal{A}^{i}$ by allowing a player $m(i)$ (the single deviator in step $i$) to change their strategy and improve their utility, i.e., $U_{m(i)}(\mathcal{A}^{i+1})<U_{m(i)}(\mathcal{A}^{i})$. Furthermore, a game is said to have the FIP property if every improvement path in the game is of finite length.
\end{definition}

\par The FIP property implies that any best-response correspondence can reach a Nash equilibrium in a finite number of iterations~\cite{potential}. Definition~\ref{def:def_bestresponse} provides the formal definition of the best-response correspondence.

\begin{definition}
	\label{def:best-response}
    For each player $m\in\mathcal{M}$, their best-response correspondence corresponds to a set-valued mapping $\mathbf{B}_m(a_{-m})$: $\mathbf{A}_{-m}\longmapsto \mathbf{A}_{m}$ such that
    \begin{sequation}
    \label{eq.bestresponse}
        \mathbf{B}_m(a_{-m})=\big\{a_m^* \mid a_m^* \in \underset{a_m \in \mathbf{A}_m}{\arg \min } U_m\left(a_m, a_{-m}\right)\big\} .
    \end{sequation}
    \label{def:def_bestresponse}
\end{definition}

\par Therefore, by demonstrating that the MUD-TOG is an exact potential game, we can obtain a Nash equilibrium solution for it. The proof for this is provided in Theorem~\ref{theorem-PG}.

%
%

\begin{theorem}
\par The MUD-TOG is an exact potential game with the following potential function
\begin{sequation}
\label{eq.PF}
\begin{aligned}
  &F(\mathcal{A})=\sum_{i\in \mathcal{M}}\sum_{n\in \mathcal{N}}I_{\{a_i=n\}}\left(\frac{Q_n^\mathrm{c}}{V}E_{n,i}^{\mathrm{c}}+\beta_{n,i}\sum_{j\leq i}I_{\{a_j=n\}}\beta_{n,j}+\right. \\
  &\left.\phi_{n,i}\sum_{j\leq i}I_{\{a_j=n\}}\phi_{n,j}\right)+\sum_{i\in \mathcal{M}}I_{\{a_i=u\}}\left(\beta_{u,i}\sum_{j\leq i}I_{\{a_j=u\}}\beta_{u,j}\right.\\
  &\left.+\phi_{u,i}\sum_{j\leq i}I_{\{a_j=u\}}\phi_{u,j}\right)+\sum_{i\in \mathcal{M}}I_{\{a_i=0\}}U_i^{\mathrm{loc}},\ \forall j\in \mathcal{M},
\end{aligned}
\end{sequation}

\noindent where $\beta_{s,i}=\sqrt{\frac{\gamma_i\eta_iD_i}{F_s^{\mathrm{max}}}}$, $\phi_{s,i}=\sqrt{\frac{\gamma_iD_i+(1-\gamma_i)p_iD_i}{r_{s,i}}}$ and $s\in \mathcal{S}$.
\label{theorem-PG}
\end{theorem}

%
%
\begin{proof}
Refer to Section IV of the supplementary material.
\end{proof}

\par Finally, let us consider the effect of constraint (\ref{Pc}) on the game. We can infer that imposing the constraint may render some strategy profiles infeasible. Suppose $\mathbb{A}^{\prime}$ is the feasible strategy space, this leads to a new game $\Gamma^{\prime}=\{\mathcal{M},\mathbb{A}^{\prime}, (U_m)_{m\in \mathcal{M}}\}$. Theorem~\ref{theo:theo_cons} demonstrates that the game $\Gamma^{\prime}$ is also an exact potential game.

%
%

\begin{theorem}
\label{theo:theo_cons}
$\Gamma^{\prime}$ possesses the same potential function as $\Gamma$, which also is an exact potential game.
\end{theorem}

%
%
\begin{proof}
Refer to Section V of the supplementary material.
\end{proof}

\par The main idea for solving the MUD-TOG is to iteratively update the offloading strategies of the players using the best-response correspondence until a Nash equilibrium is reached. This process is outlined in Algorithm \ref{Algorithm 1}, and the main steps are described as follows. \textbf{i)} All UDs choose local computing for the initial setting (Line 1). \textbf{ii)} During each iteration, the process is partitioned into $N$ decision slots (Lines 4-15). In each slot, one UD is chosen to update its offloading decision based on the best-response correspondence, while the offloading decisions of the remaining UDs remain unchanged. \textbf{iii)} If a lower utility is achieved and constraint (\ref{Pc}) is satisfied, the offloading decision of the UD is updated. Otherwise, the offloading decision remains unchanged. \textbf{iv)} The MUD-TOG achieves a Nash equilibrium when no UD modifies its offloading decision. 

\begin{algorithm}	
	\label{Algorithm 1}
	\SetAlgoLined
	\KwIn{The UD information $\{\mathbf{St}_m^{\mathrm{UD}}(t)\}_{m\in \mathcal{M}}$ and the current location $\{\mathbf{q}_n(t)\}_{n\in\mathcal{N}}$ of SUAVs.}
	\KwOut{The task offloading and resource allocation decisions $\{\mathcal{A}^{*},\mathcal{F}^{*},\mathcal{W}^{*}\}$.}
	\textbf{ Initialization:} 
	The iteration number $l=1$, $\mathcal{A}^0=\emptyset$ and $\mathcal{A}^1=\{0,\dots,0\}$\;
	\Repeat{$\mathcal{A}^{l-1} = \mathcal{A}^l$}
	{
		$\mathcal{A}^{l-1}=\mathcal{A}^{l}$\;
		\For{\text{UD} $m\in \mathcal{M}$}
		{
			$\mathbf{A}^{\prime}_m=\emptyset$\;
			\For{$a_m\in \mathbf{A}_m$}
			{
				Obtain $F_m^{*}$ and $w_m^{*}$ based on \eqref{eq.solution}\;
				Calculate $T_m$ based on (\ref{eq.delay})\;
				Calculate $U_m$ based on (\ref{eq.utility})\;
				\If{$T_m\leq T_m^{\mathrm{max}}$}
				{
					$\mathbf{A}^{\prime}_m=\mathbf{A}^{\prime}_m\cup\{a_m\}$\;
				}
				Calculate $\mathbf{B}_m(a_{-m})$ from $\mathbf{A}^{\prime}_m$ based on Eq. \eqref{eq.bestresponse}\;
				Randomly select a $a_m^{*}$ from $\mathbf{B}_m(a_{-m})$\;
			}
			$\mathcal{A}^{l}(m) = a_m^{*}$\;
		}
		Update $l=l+1$\;
	}
	$\mathcal{A}^{*}=\mathcal{A}^{l}$\;
	Obtain $\{\mathcal{F}^{*},\mathcal{W}^{*}\}$ based on Eq. \eqref{eq.solution}\;
	\Return{$\{\mathcal{A}^{*},\mathcal{F}^{*},\mathcal{W}^{*}\}$.}
	\caption{The First Stage Algorithm}
\end{algorithm}

\par \textbf{(3) Analysis of Nash Equilibrium.} Based on the above analysis, a Nash equilibrium can be obtained by the proposed algorithm. Next, we first analyze the computational complexity, which is presented in Theorem \ref{theo:Nash-complexity}.

%
%

\begin{theorem}
	\label{theo:Nash-complexity}
	\par The number of iterations $I_c$ required for the proposed algorithm to converge to a Nash equilibrium is upper bounded as follows:
	\begin{equation}
		I_c\leq \frac{M^2(\beta_{\text{max}}^2+\phi_{\text{max}}^2)+M(Q_{\text{max}}+U_{\text{max}}^{\text{loc}})}{\epsilon_{\text{min}}}.
	\end{equation}
\end{theorem}

\begin{proof}
	Refer to Section VI of the supplementary material.
\end{proof}

\par To evaluate the performance of the equilibrium solution, the Price of Anarchy (PoA) is introduced to quantify the gap between the worst-case equilibrium and the centralized optimal solutions, which can provide a bound on the sub-optimality of our proposed algorithm. Mathematically, PoA can be defined as
\begin{equation}
	\label{eq.def_poa}
	\mathrm{PoA}=\frac{\max _{\mathcal{A} \in \mathbf{A}^*} \sum_{m\in\mathcal{M}}U_m(\mathcal{A})}{\min _{\mathcal{A} \in \mathbb{A}} \sum_{m\in\mathcal{M}}U_m(\mathcal{A})},
\end{equation}

\noindent where $\mathbf{A}^*$ denotes the set of Nash equilibrium for the MUD-TOG. Clearly, a smaller PoA indicates better performance of the obtained equilibrium solution. In the following, we illustrate the bound of the PoA through Theorem \ref{the.PoA}.

%
%

\begin{theorem}
\par Let $a_m=m$ indicate the local computing, $U(\mathcal{A})$ denote the objective function of problem $\textbf{P1.2}$, i.e., $U(\mathcal{A})=\sum_{m\in\mathcal{M}}U_m(\mathcal{A})$, $\mathcal{A}^*$ indicates a Nash equilibrium, and $\hat{\mathcal{A}}$ denotes the centralized optimal strategy. For the formulated MUD-TOG, the PoA defined in Eq.~\eqref{eq.def_poa} satisfies
\label{the.PoA}
\begin{equation}
	1\leq \text{PoA} \leq 3 - \frac{\sum_{i\in\mathcal{I}}\big(G_i(\mathcal{A}^*)+G_i(\hat{\mathcal{A})}\big)}{U(\hat{\mathcal{A}})},
\end{equation}

\noindent where $\mathcal{I}=\mathcal{M}\cup\mathcal{N}\cup\{u\}$ represents the task offloading strategy space of all UDs, $G_i(\mathcal{A})=\sum_{m\in\mathbf{M}_i(\mathcal{A})}G_{i,m}$, and $\mathbf{M}_i(\mathcal{A})=\{m|a_m=i,\forall m\in\mathcal{M}\}$ denotes the set of UDs whose offloading decision satisfies $a_m=i$. Furthermore, $G_{i,m}=Q_i^{\text{c}}E_{i,m}^{\text{c}}/V$ if $i\in\mathcal{N}$; otherwise, $G_{i,m}=0$.
\end{theorem}

%
%

\begin{proof}
	Refer to Section VII of the supplementary material.
\end{proof}

%
%

\subsection{Stage 2: SUAV Trajectory Control}
\label{subsec:UAV Movement}

\par Given the task offloading decisions $\mathcal{A}^{*}$ and resource allocation $\{\mathcal{F}^{*},\mathcal{W}^{*}\}$, while removing irrelevant constant terms, problem $\textbf{P}^{\prime}$ can be converted into subproblem $\mathbf{P2}$ to decide trajectory control for SUAVs, which is given as follows:

\vspace{-0.8em}
{\small \begin{align}
    &\mathbf{P2}:\underset{\mathcal{Q}}{\text{min}}\ \sum_{n\in \mathcal{N}}\sum_{m\in \mathbf{M}_n(\mathcal{A})}V\frac{\gamma_m^{\text{T}}D_m+\gamma_m^{\text{E}}p_mD_m}{w_{n,m}^{*}B_n\log_2(1+\frac{\phi_{n,m}}{\|\mathbf{q}_{n^\prime}-\mathbf{q}_m\|^2+H^2})}+\notag \\
    &\sum_{n\in \mathcal{N}}Q_n^{\text{p}}\left(C_1\big(1+\frac{3 v_n^2}{U_{\text {p}}^2}\big)+C_2\sqrt{\sqrt{C_3+\frac{v_n^4}{4}}-\frac{v_n^2}{2}}+C_4v_n^3\right)\tau \label{P2}\\
    &\text{s.t.}\ (\ref{Ph}), \notag
\end{align}}

\noindent where $\mathbf{q}_{n^\prime}=\mathbf{q}_{n}(t+1)$, $\mathbf{q}_n=\mathbf{q}_{n}(t)$, $v_n=\frac{\|\mathbf{q}_{n^\prime}-\mathbf{q}_{n}\|}{\Delta t}$, and $\phi_{n,m} = \frac{p_m c^2[\rho_{n,m}^{\text{LoS}}|h_{n,m}^{\text{LoS}}|^2\eta^{\text{NLoS}}+(1-\rho_{n,m}^{\text{LoS}})|h_{n,m}^{\text{NLoS}}|^2\eta^{\text{LoS}}]}{(4\pi f_n)^2\varpi_0\eta^{\text{LoS}}\eta^{\text{NLoS}}}$. Clearly, the objective function in (\ref{P2}) is non-convex with respect to $\mathbf{q}_{n^\prime}$ due to the following non-convex terms
\begin{sequation}
\label{eq.non-convex-terms}
    \begin{cases}
    TMR_n=\sqrt{\sqrt{C_3+\frac{v_n^4}{4}}-\frac{v_n^2}{2}},\ \forall n\in\mathcal{N},\\
    TML_{n,m}=\frac{1}{\log_2\big(1+\frac{\phi_{n,m}}{\|\mathbf{q}_{n^\prime}-\mathbf{q}_m\|^2+H^2}\big)},\ \forall n\in\mathcal{N},m\in\mathbf{M}_n(\mathcal{A}). 
    \end{cases}
\end{sequation}

\noindent Therefore, solving problem $\textbf{P2}$ directly is difficult. Next, we transform the objective function into a convex function by introducing slack variables.

\par For the non-convex term $TMR_n$, we introduce the slack variable $\xi_n$ such that $\xi_n=TMR_n$ and add the following constraint:
\begin{sequation}
\label{eq.slack1}
    \xi_n \geq \sqrt{\sqrt{C_3+\frac{v_n^4}{4}}-\frac{v_n^2}{2}} \Longrightarrow \frac{C_3}{\xi_n^2} \leq \xi_n^2+v_n^2,\ \forall n\in\mathcal{N}.
\end{sequation}

\par For the non-convex term $TML_{n,m}$, we introduce the slack variable $\zeta_{n,m}$ such that $1/\zeta_{n,m}=TML_{n,m}$ and add the following constraint:
\begin{sequation}
\label{eq.slack2}
    \zeta_{n,m} \leq \log _2\big(1+\frac{\phi_{n,m}}{H^2+\left\|\mathbf{q}_{n^{\prime}}-\mathbf{q}_m\right\|^2}\big),\forall n\in\mathcal{N},m\in\mathbf{M}_n(\mathcal{A}).
\end{sequation}

\par According to the above-mentioned relaxation transformation, problem $\textbf{P2}$ can be equivalently transformed as follows:
\begin{align}
    \mathbf{P2}^{\prime}:&\ \underset{\mathcal{Q},\mathbf{\zeta},\mathbf{\xi}}{\text{min}}\ V \sum_{n\in\mathcal{N}}\sum_{m\in \mathbf{M}_n(\mathcal{A})}\frac{\gamma_m^{\text{T}}D_m+\gamma_m^{\text{E}}p_mD_m}{w_{n,m}^{*}B_n\zeta_{n,m}}\notag\\
    &+\sum_{n\in\mathcal{N}}Q_n^{\text{P}}\big(C_1\big(1+\frac{3 v_n^2}{U_{\text {p }}^2}\big)+C_2\xi_n+C_4v_n^3\big)\tau\label{P2_1}\\
    \text{s.t.} \ &(\ref{Ph}),(\ref{eq.slack1})\ and\ (\ref{eq.slack2}), \notag
\end{align}

\noindent where $\mathcal{Q}=\{\mathbf{q}_{n^\prime}\}_{n\in\mathcal{N}}$, $\mathbf{\zeta}=\{\zeta_{n,m}\}_{n\in\mathcal{N},m\in\mathbf{M}_n(\mathcal{A})}$ and $\mathbf{\xi}=\{\xi_n\}_{n\in\mathcal{N}}$.

\par Note that problem $\mathbf{P2}^{\prime}$ is equivalent to problem $\mathbf{P2}$, which can be demonstrated by Theorem \ref{the:the5}.

%
%

\begin{theorem}
\label{the:the5}
Problem $\mathbf{P2}^{\prime}$ is equivalent to problem $\mathbf{P2}$.
\end{theorem}

%
%

\begin{proof}
\label{pro:equ}
  Refer to Section VIII of the supplementary material.
\end{proof}

\par For problem $\mathbf{P2}^{\prime}$, the optimization objective (\ref{P2_1}) is convex but constraints (\ref{eq.slack1}), (\ref{eq.slack2}) and (\ref{cons2}) are still non-convex. Similar to~\cite{Yang2022,UAV-H}, the successive convex approximation (SCA) method is adopted to solve the non-convexity of above constraints, as demonstrated in the Theorems~\ref{pro:pro5-2-1}, ~\ref{pro:pro5-2-2}, and~\ref{pro:pro5-2-3}.

%
%

\begin{theorem}
\label{pro:pro5-2-1}
For constraint (\ref{eq.slack1}), let $f_n(\mathbf{q}_{n^\prime},\xi_n)=\xi_n^2+v_n^2$ and given a local point $\mathbf{q}^{(l)}_{n^\prime}$ at the $l$-th iteration, a global concave lower bound of $f_n(\mathbf{q}_{n^\prime},\xi_n)$ can be obtained as

\vspace{-0.8em}
{\small \begin{align}
    \label{eq.pro5-2-1}
    f_n^{(l)}(\mathbf{q}_{n^\prime},\xi_n) \triangleq&\left(\xi_n^{(l)}\right)^2+2 \xi_n^{(l)}\left(\xi_n-\xi_n^{(l)}\right)+\frac{\|\mathbf{q}_{n^{\prime}}^{(l)}-\mathbf{q}_n\|^2}{\tau^2}\notag\\
    &+\frac{2}{\tau^2}(\mathbf{q}_{n^{\prime}}^{(l)}-\mathbf{q}_n)^T\left(\mathbf{q}_{n^{\prime}}-\mathbf{q}_n\right),
\end{align}}

\noindent where $\xi_m^{(l)}$ is defined as
\begin{sequation}
    \label{eq.y-l}
    \xi_n^{(l)}=\sqrt{\sqrt{C_3+\frac{\|\mathbf{q}_{n^{\prime}}^{(l)}-\mathbf{q}_n\|^4}{4 \tau^4}}-\frac{\|\mathbf{q}_{n^{\prime}}^{(l)}-\mathbf{q}_n\|^2}{2 \tau^2}} .
\end{sequation}
\end{theorem}

%
%

\begin{proof}
Refer to Section IX of the supplementary material.
\end{proof}

%
%

\begin{theorem}
\label{pro:pro5-2-2}

\par For constraint (\ref{eq.slack2}), let $g_{n,m}(\mathbf{q}_{n^\prime})=\log _2\big(1+\frac{\phi_{n,m}}{H^2+\left\|\mathbf{q}_{n^{\prime}}-\mathbf{q}_m\right\|^2}\big)$ and given a local point $\mathbf{q}^{(l)}_{n^\prime}$ at the $l$-th iteration, a global concave lower bound of $g_{n,m}(\mathbf{q}_{n^\prime})$ can be obtained as follows:

\vspace{-0.8em}
{\small \begin{align}
    \label{eq.taylor2}
    &g_{n,m}^{(l)}(\mathbf{q}_{n^\prime}) \triangleq \log _2\big(1+\frac{\phi_{n,m}}{H^2+\|\mathbf{q}_{n^{\prime}}^{(l)}-\mathbf{q}_m\|^2}\big)\notag \\
    & -\frac{\phi_{n,m} (\log_2 e)(\|\mathbf{q}_{n^{\prime}}-\mathbf{q}_m\|^2-\|\mathbf{q}_{n^{\prime}}^{(l)}-\mathbf{q}_m\|^2)}{[\phi_{n,m}+(H^2+\|\mathbf{q}_{n^{\prime}}^{(l)}-\mathbf{q}_m\|^2)](H^2+\|\mathbf{q}_{n^{\prime}}^{(l)}-\mathbf{q}_m\|^2)}.
\end{align}}
\end{theorem}

%
%

\begin{proof}
	Refer to Section X of the supplementary material.
\end{proof}

%
%

\begin{theorem}
\label{pro:pro5-2-3}

\par For constraint (\ref{cons2}), let $h_{i,j}(\mathbf{q}_{i^\prime},\mathbf{q}_{j^\prime})=\|\mathbf{q}_{i^\prime}-\mathbf{q}_{j^\prime}\|^2$ and given a local point $(\mathbf{q}_{i^\prime}^{(l)},\mathbf{q}_{j^\prime}^{(l)})$ at the $l$-th iteration, a global concave lower bound of $h_{i,j}(\mathbf{q}_{i^\prime},\mathbf{q}_{j^\prime})$ can be obtained as follows:
\begin{align}
    \label{eq.pro5-2-3}
    h_{i,j}^{(l)}(\mathbf{q}_{i^\prime},\mathbf{q}_{j^\prime})\triangleq&2(\mathbf{q}_{i^\prime}^{(l)}-\mathbf{q}_{j^\prime}^{(l)})^{T}\notag \left[(\mathbf{q}_{i^\prime}-\mathbf{q}_{i^\prime}^{(l)})-(\mathbf{q}_{j^\prime}-\mathbf{q}_{j^\prime}^{(l)})\right]\\
    &+\|\mathbf{q}_{i^\prime}^{(l)}-\mathbf{q}_{j^\prime}^{(l)}\|^2.
\end{align}
\end{theorem}

%
%

\begin{proof}
	Refer to Section XI of the supplementary material. 
\end{proof}

\par According to Theorems \ref{pro:pro5-2-1}, \ref{pro:pro5-2-2}, and \ref{pro:pro5-2-3}, at the $l$-th iteration, constraints (\ref{eq.slack1}), (\ref{eq.slack2}), and (\ref{cons2}) can be approximated as

{\small \begin{align}
    \label{eq.cons1}
    &\frac{C_3}{\xi_n^2}\leq f_n^{(l)}(\mathbf{q}_{n^\prime},\xi_n),\\
    &\zeta_{n,m}\leq g_{n,m}^{(l)}(\mathbf{q}_{n^\prime}),\\
    &(d^{\text{min}})^2\leq h_{i,j}^{(l)}(\mathbf{q}_{i^\prime},\mathbf{q}_{j^\prime}),
\end{align}}

\noindent which are convex. Therefore, problem $\mathbf{P2}^{\prime}$ is converted into a convex optimization problem, which can be efficiently resolved by off-the-shelf optimization tools like CVX~\cite{cvx}. We summarize the second stage algorithm in Algorithm \ref{Algorithm 2}. 

\vspace{-0.5em}
\begin{algorithm}
	\label{Algorithm 2}
	\SetAlgoLined
	\KwIn{The optimal task offloading and resource allocation decisions $\{\mathcal{A}^{*},\mathcal{F}^{*},\mathcal{W}^{*}\}$.}
	\KwOut{The next location $\{\mathbf{q}_{n^\prime}\}_{n\in\mathcal{N}}$.}
	\textbf{ Initialization:}
	The accuracy threshold $\varepsilon = 0.01$, the iterative number $l=1$, the local point $\mathbf{q}^{(0)}_{n^\prime}=\mathbf{q}_{n}$ and the objective function value $G^{(0)}=0$\;
	\Repeat{$|G^{(l)}-G^{(l-1)}|<\varepsilon$}{
		Calculate $\{\xi_n^{(l)}\}_{n\in\mathcal{N}}$ based on Eq.~\eqref{eq.y-l}\;
		Obtain the optimal position $\{\mathbf{q}^{*}_{n^\prime}\}_{n\in\mathcal{N}}$ and the objective value $G^{(l)}$ by solving problem $\mathbf{P2}^{\prime}$\;
		Update the local point $\mathbf{q}^{(l)}_{n^\prime}=\mathbf{q}^{*}_{n^\prime}$\;
		Update $l=l+1$\;}
	\Return{$\{\mathbf{q}^{*}_{n^\prime}\}_{n\in\mathcal{N}}$.}
	\vspace{-1pt}
	\caption{The Second Stage Algorithm}
\end{algorithm}
\vspace{-0.5em}

%
%

\subsection{Main Steps of OJTRTA and Performance Analysis}
\label{sec:Performance Analysis}

\par In this section, the main steps of OJTRTA are described in Algorithm \ref{Algorithm 3}, and the corresponding analyses are provided in Theorems~\ref{the:performance}, \ref{theorem-energy}, and \ref{the:complexity}. 

\par Specifically, at the beginning of each time slot, UDs first upload their task request information (line 3). Based on the collected UD information, the proposed approach derives the strategies of resource allocation and task offloading by invoking Algorithm 1 (line 4). Given these strategies, the proposed approach then determines the trajectory control for the SUAVs by invoking Algorithm 2 (line 5). Subsequently, the UDs execute their computational tasks according to the established offloading strategies, while the SUAVs allocate computing and communication resources to provide edge computing services and move to the next location (lines 6 and 7). Finally, the system cost and energy consumption queues are calculated and updated (lines 8 to 10).

\vspace{-0.5em}
\begin{algorithm}
    \label{Algorithm 3}
    \SetAlgoLined
    \KwIn{The energy queue $Q_n^{\text{c}}(1)=0$, $Q_n^{\text{p}}(1)=0$ and the control parameter $V$.}
    \KwOut{Time-averaged UD cost $TAC$.}
    \textbf{ Initialization:} 
    Initialize $TAC = 0$ and the initial position of SUAVs $\mathbf{q}_n(1)=\mathbf{q}_n^{\mathrm{ini}}$\;
    \For{$t=1$ to $t=T$}
    {
        Acquire the UD information $\{\mathbf{St}_m^{\text{UD}}(t)\}_{m\in \mathcal{M}}$\; 
        With fixed $\{\mathbf{q}_n(t)\}_{n\in\mathcal{N}}$, call Algorithm \ref{Algorithm 1} to obtain $\{\mathcal{A}^{*},\mathcal{F}^{*},\mathcal{W}^{*}\}$\;
        With fixed $\{\mathcal{A}^{*},\mathcal{F}^{*},\mathcal{W}^{*}\}$, call Algorithm \ref{Algorithm 2} to obtain $\{\mathbf{q}_{n^\prime}^{*}\}_{n\in\mathcal{N}}$\;
        All UDs perform their tasks based on $\mathcal{A}^{*}$ and obtain corresponding cost $C_m^{*}(t)$\;
        The SUAVs provides MEC service to the UDs and flies towards position $\{\mathbf{q}^{*}_{n^\prime}\}_{n\in\mathcal{N}}$\;
        System cost $C_s(t)=\sum_{m=1}^MC_m^{*}(t)$\;
        $TAC=TAC+C_s(t)$\;
        Update the energy queue $Q_n^{\mathrm{c}}(t+1)$ and $Q_n^{\mathrm{p}}(t+1)$ according to Eq.~\eqref{eq.queue}\;
        Update $t=t+1$\;
    }
    $TAC=TAC/T$\;
    \Return{$TAC$.}
    \caption{OJTRTA}
\end{algorithm}
\vspace{-1em}

%
%

\begin{theorem}
\label{the:performance}

\par Assume that the proposed algorithm results in an optimality gap $C\geq 0$ in solving $\mathbf{P}^{\prime}$ and $C_s^{\mathrm{opt}}$ denotes the optimal time-averaged UD cost that problem $\mathbf{P}$ can achieve over all policies given full knowledge of the future computing demands and locations for all UDs, the time-averaged UD cost achieved by the proposed algorithm is bounded by
\begin{sequation}
    \label{eq.performance}
    \frac{1}{T} \sum_{t=1}^T\sum_{m=1}^M C_m(t) \leq C_s^{\mathrm{opt}}+\frac{WT+C }{V},
\end{sequation}

\noindent where $W$ is defined in Theorem~\ref{the:drift-plus-penalty}.
\end{theorem}

%
%

\begin{proof}
	Refer to Section XII of the supplementary material.
\end{proof}

%
%

\begin{theorem}
\label{theorem-energy}

\par The proposed algorithm can satisfy the SUAV energy consumption constraint defined in (\ref{eq.eng-cons}).
\end{theorem}

%
%

\begin{proof}
Refer to Section XIII of the supplementary material.
\end{proof}

%
%

\begin{theorem}
\label{the:complexity}

\par The proposed OJTRTA has a polynomial worst-case complexity in each time slot, i.e., $\mathcal{O}(I_cMN+(N+M)^{2.5}(N^2+M)\log_2(\frac{1}{\varepsilon}))$, where $M$ represents the number of UDs, $I_c$ denotes the number of iterations necessary for Algorithm \ref{Algorithm 1} to reach convergence at the Nash equilibrium, and $\varepsilon$ is the accuracy of SCA for solving problem $\mathbf{P2}^{\prime}$.

\end{theorem}

%
%

\begin{proof}
Refer to Section XIV of the supplementary material.
\end{proof}

\par Note that the proposed approach strikes a trade-off between computational complexity and performance to ensure feasibility, which is reasonable. Specifically, for the complex JTRTOP, obtaining the optimal solution usually incurs a large computational overhead or is even impossible, making it unsuitable for the real-time decision-making scenario considered. In contrast, the proposed approach reduces computational complexity at the expense of some performance trade-offs, which is realistic since the complexity reduction often comes with performance degradation~\cite{Wu2020}. Despite these trade-offs, the proposed approach enables real-time decision-making and achieves a feasible solution that satisfies the requirements of UDs while meeting the constraints of the system. Moreover, the simulation results will demonstrate the superior performance of the proposed approach.

%
%
\section{Simulation Results}
\label{sec:Simulation Results and Analysis}

\par In this section, we conduct simulation experiments to validate the effectiveness of the proposed OJTRTA.

%
%

\subsection{Simulation Setup}
\label{subsec:Simulation setups}

\begin{table}
	\setlength{\abovecaptionskip}{-1em}%
	\setlength{\belowcaptionskip}{0pt}%
	\caption{Simulation Parameters}
	\label{parameters}
	\renewcommand*{\arraystretch}{1.15}
	\begin{center}
		\begin{tabular}{p{.06\textwidth}<{\centering}|p{.21\textwidth}|p{.13\textwidth}}
			\hline
			\hline
			\textbf{Symbol}&\textbf{Meaning}&\textbf{Value (Unit)}\\
			\hline
			$D_m$ &Task size &$[0.2,1]$ Mb\\
			\hline
			$\eta_m$ &Computation intensity of tasks &$[500,1500]$ cycles/bit \\
			\hline
			$T_m^{\text{max}}$ &Maximum tolerable delay of tasks &$1$ s\\
			\hline
			$\alpha$ &Memory level of velocity &$0.9$ \\
			\hline
			$\overline{\mathbf{v}}_m$ &the velocity of UD $m$ &$1$ m/s~\cite{Yang2022} \\
			\hline
			$\sigma_m$ & The asymptotic standard deviation of velocity &$2$~\cite{Yang2022} \\
			\hline
			$F_n^{\text{max}}$ &Computation resources of SUAVs &$20$ GHz\\
			\hline
			$v_n^{\text{max}}$ &Maximum flight speed of SUAVs &$25$ m/s~\cite{Yang2022}\\   
			\hline
			$d^{\text{min}}$ &Minimum safety distance &$10$ m\\
			\hline
			$F_u^{\text{max}}$ &Computation resources of LUAV &$30$ GHz\\
			\hline
			$B_s$ &Bandwidth of MEC server $s$ &$10$ MHz $(s=u)$, \newline $5$ MHz $(s\in \mathcal{N})$ \\
			\hline
			$p_m$ &Transmission power of UD $m$  &$20$ dBm \\
			\hline
			$\varpi_0$ &Noise power &$-98$ dBm\\
			\hline
			$c_1$, $c_2$ &Parameters for LoS probability &10, 0.6~\cite{HouraniSL14}\\
			\hline
			$\eta^{\text{LoS}},\eta^{\text{nLoS}}$ &Additional losses for LoS and NLoS links &1.0 dB, 20 dB~\cite{HouraniSL14}\\
			\hline
			$\kappa$ &CPU parameters &$10^{-28}$\\
			\hline
			$\varpi$  &Energy consumption per unit CPU cycle of SUAVs&$8.2\times10^{-9}$ J~\cite{JiangDXI23} \\
			\hline
			$U_{\text{p}}$ &Tip speed of the rotor&120 m/s \\
			\hline
			$\gamma_m^{\text{T}}$, $\gamma_m^{\text{E}}$ &The weight coefficients of task completion delay and energy consumption for UD $m$ & 0.7, 0.3\\
			\hline
		\end{tabular}
	\end{center}
\end{table}


\subsubsection{Scenario}
\label{subsubsec:Scenario}

\par We consider a multiple-UAV-assisted MEC system, where one LUAV and four SUAVs are deployed to collaboratively provide edge computing services for $60$ ground UDs in a $1,000\times1,000\ \text{m}^2$ rectangular service area. The system timeline is segmented into $100$ equal time slots, and the duration of each time slot is set to $\Delta t=1\ \text{s}$.

\subsubsection{Parameters}
\label{subsubsec:Parameters}

\par For the LUAV, the fixed horizontal position and altitude are set as $\mathbf{q}_u=[500, 500]$ and $H_u=300\ \text{m}$, respectively. For the SUAVs, the fixed altitude is set as $H=100\ \text{m}$, and the initial positions are set as $\mathbf{q}_1^{\text{ini}}=[100,100]$, $\mathbf{q}_2^{\text{ini}}=[100,900]$, $\mathbf{q}_3^{\text{ini}}=[900,900]$, and $\mathbf{q}_4^{\text{ini}}=[900,100]$, respectively. For the UDs, the computing capacity of each UD is randomly taken from $\{1,1.5,2\}\  \text{GHz}$. The default values for the remaining parameters can be found in Table~\ref{parameters}.

\subsubsection{Comparative Approaches}
\label{subsubsec:Comparative Approaches}

\par To validate the effectiveness of the proposed OJTRTA, this work compares OJTRTA with the following approaches.
\begin{itemize}
	
    \item \textit{Entire offloading (EO)}: All UDs offload their tasks to aerial servers for execution. This approach does not consider local computing.
    
    \item \textit{Equal resource allocation (ERA)}~\cite{Josilo}: The communication and computing resources of each edge server are equally allocated to the requested UDs.
    
    \item \textit{Fixed location deployment (FLP)}~\cite{Wang2022}: The SUAV are deployed at fixed positions within the service area to provide edge computing services.
    
    \item \textit{Only consider QoE (OCQ)}~\cite{JiangDXI23}: Ignoring the SUAV energy consumption constraint, all decisions are made only to minimize the time-averaged UD cost.
    
    \item \textit{DDPG-based joint optimization (DDPG-JO)}~\cite{10388042}: The decisions of resource allocation, task offloading, and trajectory control are decided by the DDPG algorithm.
    
    \item \textit{PPO-based joint optimization (PPO-JO)}~\cite{10418937}: The decisions of resource allocation, task offloading, and trajectory control are determined by the PPO algorithm.
    
\end{itemize}

\subsubsection{Performance Metrics}
\label{subsubsec:Performance Metrics}

\begin{figure*}[!hbt] 
	\centering
	\setlength{\abovecaptionskip}{2pt}%
	\setlength{\belowcaptionskip}{2pt}%
	\subfigure[]
	{
		\begin{minipage}[t]{0.23\linewidth}
			\centering
			\includegraphics[scale=0.31]{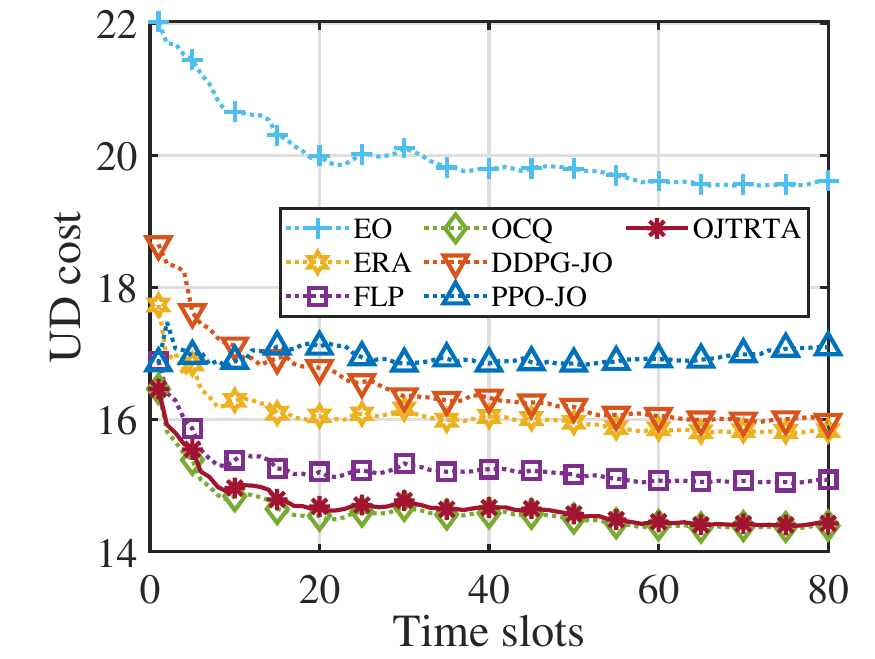}
		\end{minipage}
	}
	\subfigure[]
	{
		\begin{minipage}[t]{0.23\linewidth}
			\centering
			\includegraphics[scale=0.31]{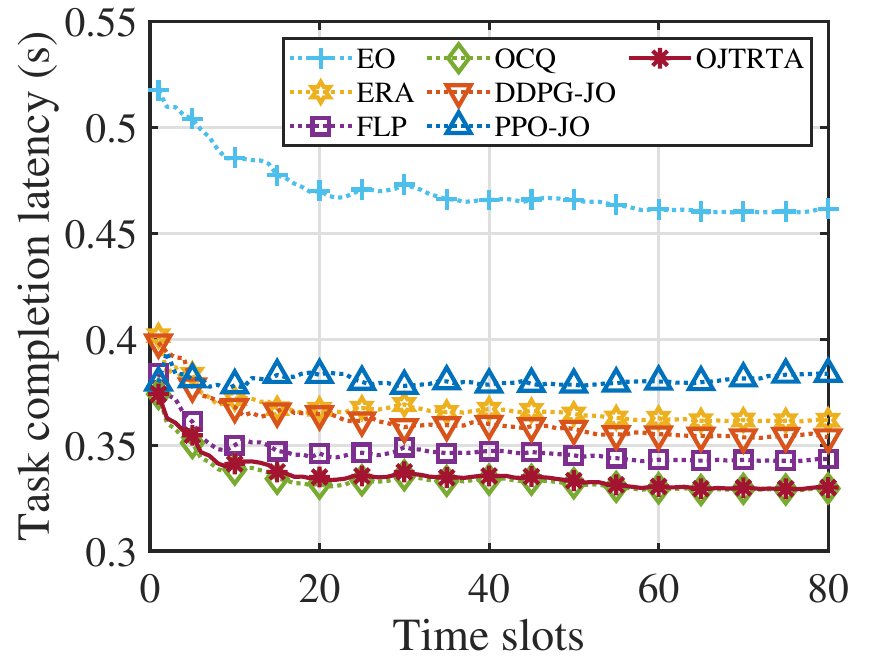}	
		\end{minipage}
	}
	\subfigure[]
	{
		\begin{minipage}[t]{0.23\linewidth}
			\centering
			\includegraphics[scale=0.31]{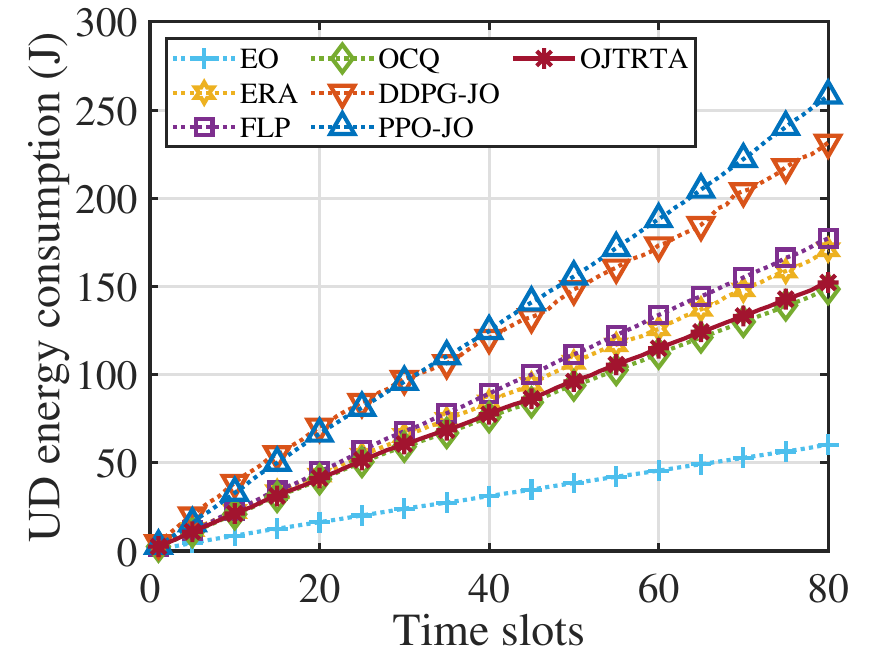}	
		\end{minipage}
	}
	\subfigure[]
	{
		\begin{minipage}[t]{0.23\linewidth}
			\centering
			\includegraphics[scale=0.31]{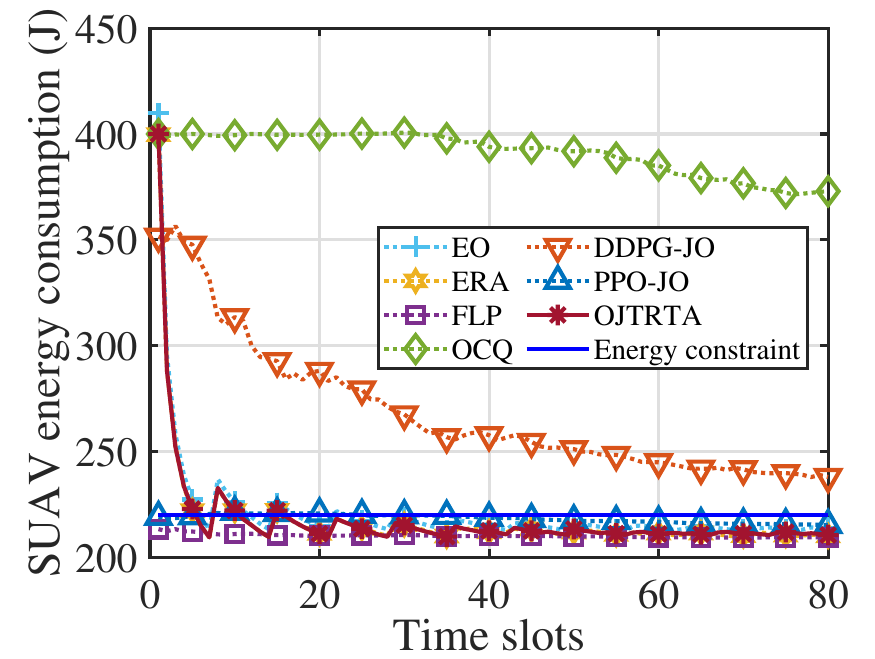}	
		\end{minipage}
	}
	\centering
	\caption{The impact of time slots on system performance. (a) Time-averaged UD cost. (b) Average task completion latency. (c) Cumulative UD energy consumption. (d) Time-averaged SUAV energy consumption.}
	\label{fig_time}
	\vspace{-1em}
\end{figure*}

\begin{figure*}[!hbt] 
	\centering
	\setlength{\abovecaptionskip}{2pt}%
	\setlength{\belowcaptionskip}{2pt}%
	\subfigure[]
	{
		\begin{minipage}[t]{0.23\linewidth}
			\centering
			\includegraphics[scale=0.31]{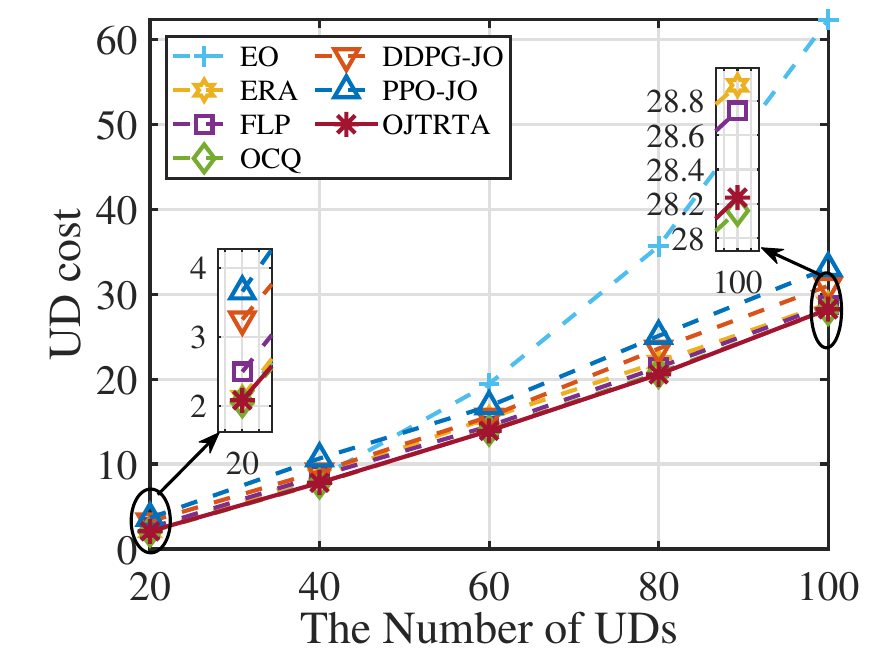}
		\end{minipage}
	}
	\subfigure[]
	{
		\begin{minipage}[t]{0.23\linewidth}
			\centering
			\includegraphics[scale=0.31]{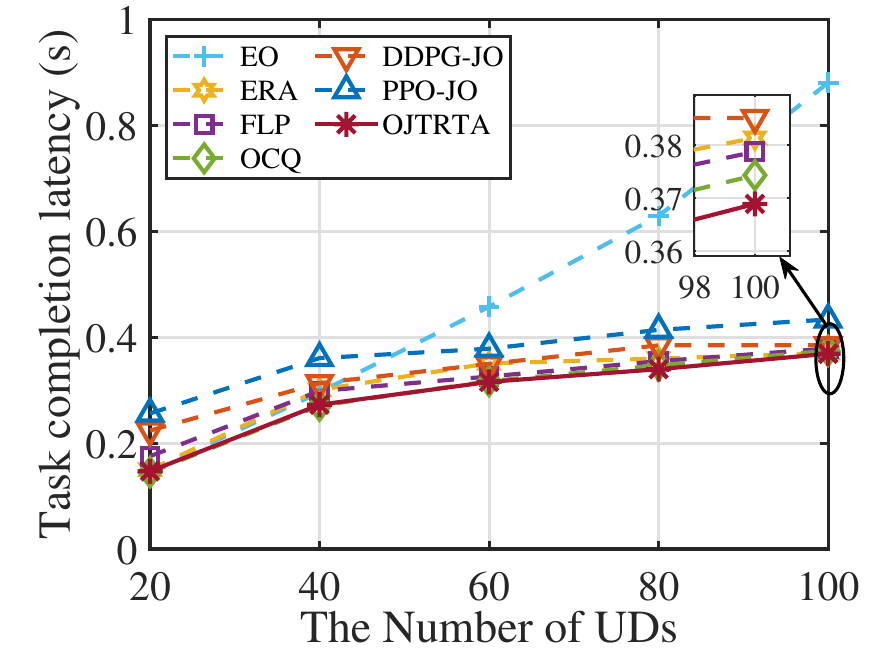}	
		\end{minipage}
	}
	\subfigure[]
	{
		\begin{minipage}[t]{0.23\linewidth}
			\centering
			\includegraphics[scale=0.31]{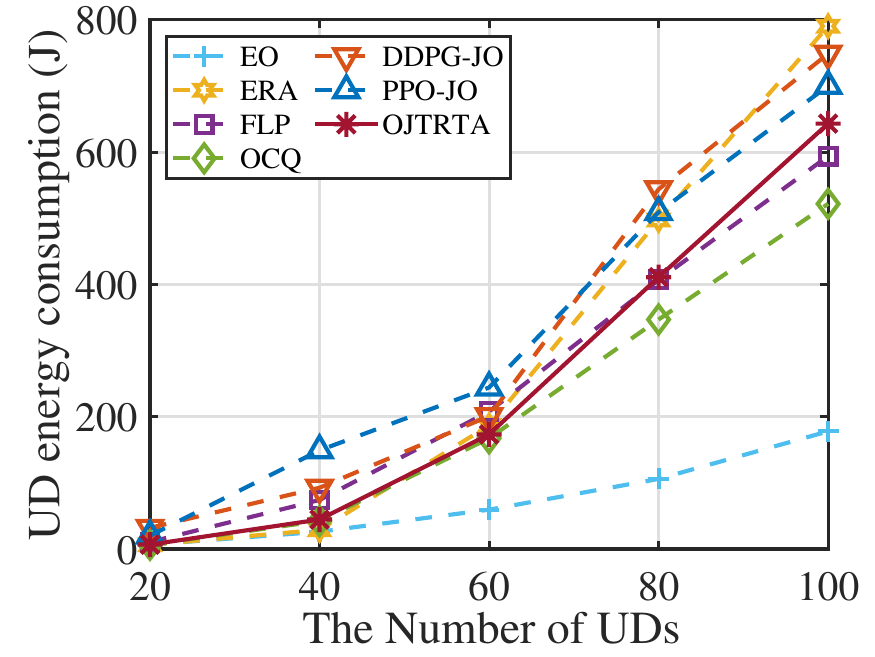}	
		\end{minipage}
	}
	\subfigure[]
	{
		\begin{minipage}[t]{0.23\linewidth}
			\centering
			\includegraphics[scale=0.31]{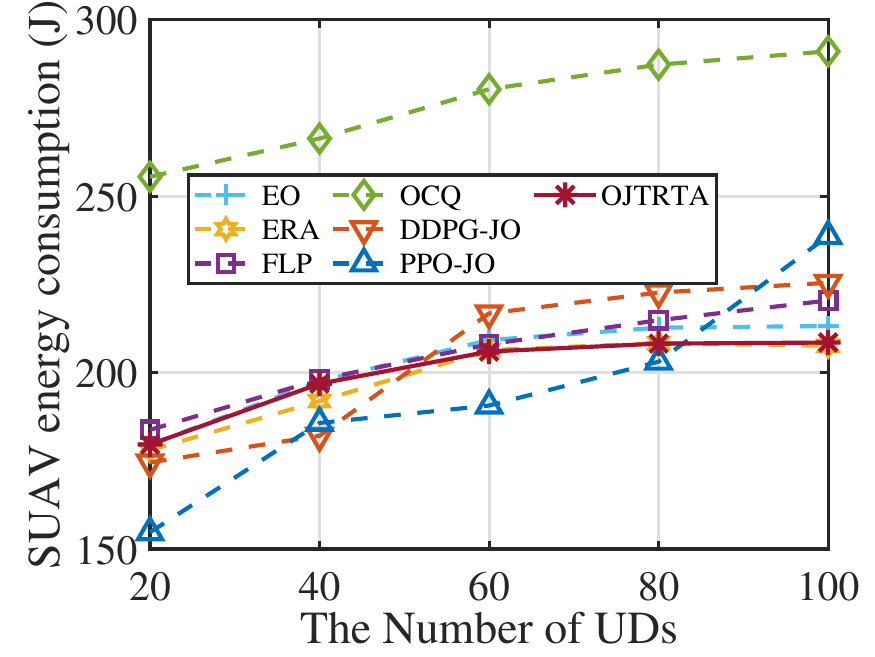}	
		\end{minipage}
	}
	\centering
	\caption{The impact of the numbers of UDs on system performance. (a) Time-averaged UD cost. (b) Average task completion latency. (c) Cumulative UD energy consumption. (d) Time-averaged SUAV energy consumption.}
	\label{fig_UDnum}
	\vspace{-1em}
\end{figure*}

\par We evaluate the overall performance of the proposed approach based on the following performance metrics. \textit{1) Time-averaged UD cost} $\frac{1}{T}\sum_{t=1}^{T}\sum_{m=1}^{M}C_m(t)$, which represents the average cumulative cost of all UDs per unit time. \textit{2) Average task completion latency} $\frac{1}{T}\sum_{t=1}^{T}\frac{1}{M}\sum_{m=1}^{M}T_m(t)$, which indicates the average latency for completing a task. \textit{3) Cumulative UD energy consumption} $\sum_{t=1}^{T}\sum_{m=1}^{M}E_m(t)$, which signifies the cumulative energy consumption of UDs over the system timeline. \textit{4) Time-averaged SUAV energy consumption} $\frac{1}{T}\sum_{t=1}^{T}\frac{1}{N}\sum_{n=1}^{N}E_n(t)$, which means the average energy consumption of each SUAV per unit time.
%
%
\subsection{Evaluation Results}
\label{subsec:Evaluation results}

\par In this section, we first evaluate the performance of the proposed OJTRTA over time with default parameters. Then, we compare the impacts of different parameters on the performance of OJTRTA in comparison with benchmark approaches.

\subsubsection{Online Offloading Performance Evaluation}
\label{subsubsec:Online Offloading Performance Evaluation}

\par Figs.~\ref{fig_time}(a), \ref{fig_time}(b), \ref{fig_time}(c), and \ref{fig_time}(d) show the dynamics of time-averaged UD cost, average task completion latency, cumulative UD energy consumption, and time-averaged SUAV energy consumption, respectively among the seven approaches. It can be observed that the time-averaged UD cost exhibits some fluctuations over time. This is primarily attributed to the time-varying computational demands of UDs. Additionally, OCQ slightly outperforms the proposed OJTRTA in terms of time-averaged UD cost, average task completion latency, and overall UD energy consumption. This is because OCQ does not consider the energy constraints of SUAVs and thus consumes more SUAV energy. Furthermore, the proposed OJTRTA outperforms EO, ERA, and FLP in terms of time-averaged UD cost and average task completion latency. This is due to the combined local and edge offloading strategy, optimal resource allocation strategy, and trajectory control strategy of OJTRTA. Moreover, compared to DRL-based approaches, i.e., DDPG-JO and PPO-JO, the proposed OJTRTA exhibits significant advantages in terms of time-averaged UD cost, average task completion delay, and cumulative UD energy consumption. This further exemplifies the effectiveness of the proposed approach. Finally, as shown in Fig. \ref{fig_time}(d), the proposed OJTRTA ensures that the long-term energy constraints of SUAVs are satisfied, which is consistent with the analysis in Theorem~\ref{theorem-energy}. In conclusion, the set of simulation results demonstrates the effectiveness of the proposed approach in achieving better overall performance while satisfying the long-term energy consumption constraints of SUAVs.

\begin{figure*}[!hbt] 
	\centering
	\setlength{\abovecaptionskip}{2pt}%
	\setlength{\belowcaptionskip}{2pt}%
	\subfigure[]
	{
		\begin{minipage}[t]{0.23\linewidth}
			\centering
			\includegraphics[scale=0.31]{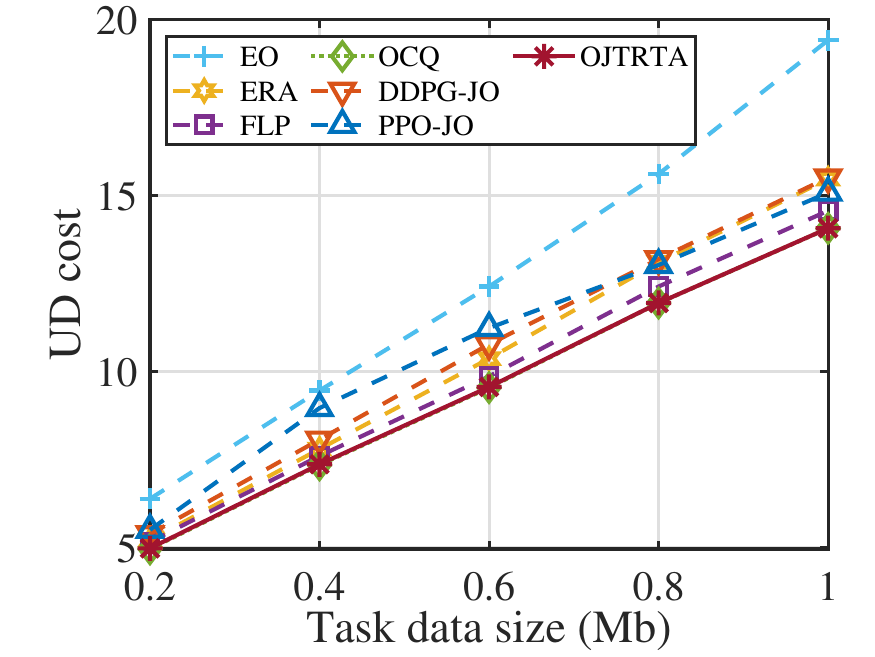}
		\end{minipage}
	}
	\subfigure[]
	{
		\begin{minipage}[t]{0.23\linewidth}
			\centering
			\includegraphics[scale=0.31]{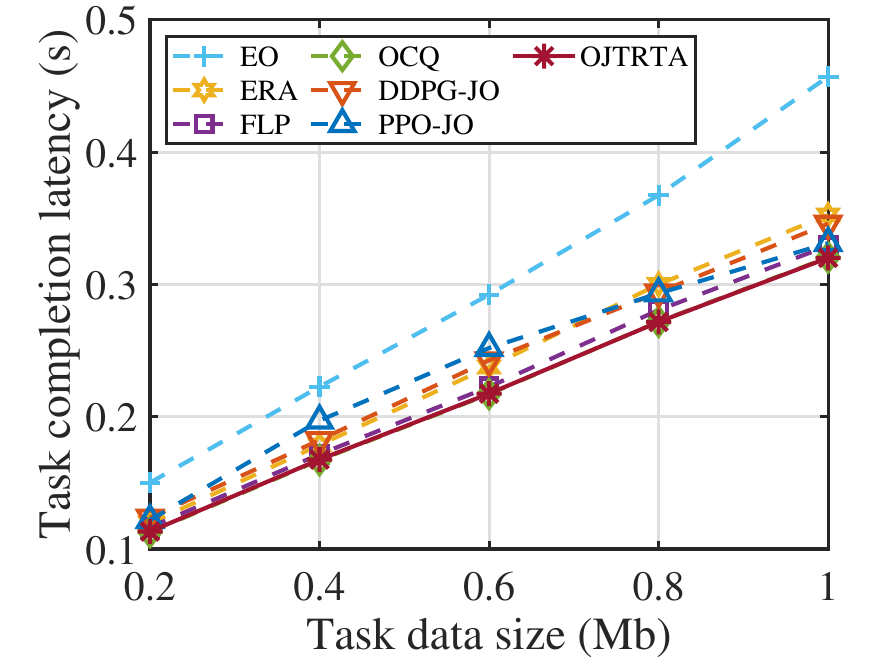}	
		\end{minipage}
	}
	\subfigure[]
	{
		\begin{minipage}[t]{0.23\linewidth}
			\centering
			\includegraphics[scale=0.31]{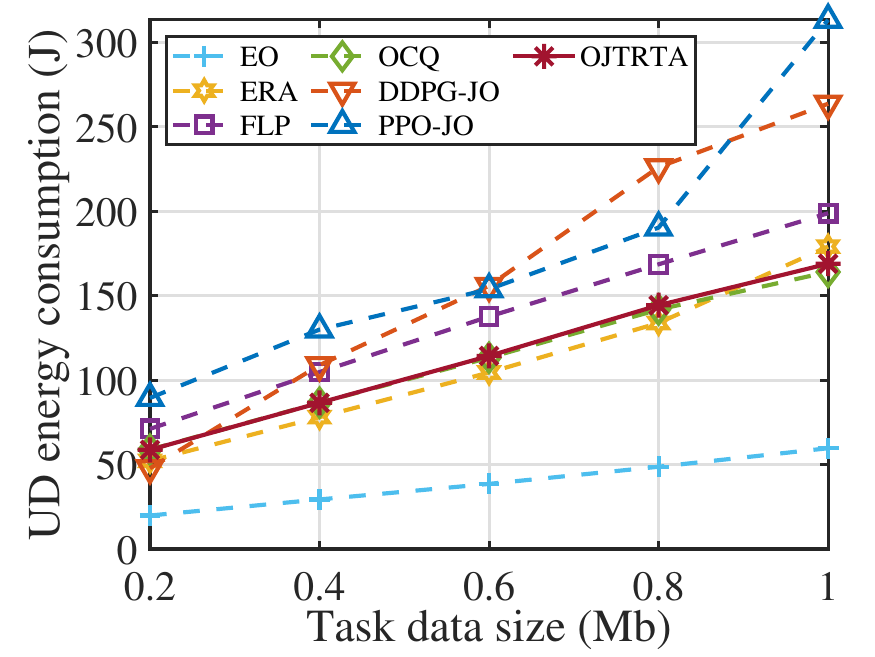}	
		\end{minipage}
	}
	\subfigure[]
	{
		\begin{minipage}[t]{0.23\linewidth}
			\centering
			\includegraphics[scale=0.31]{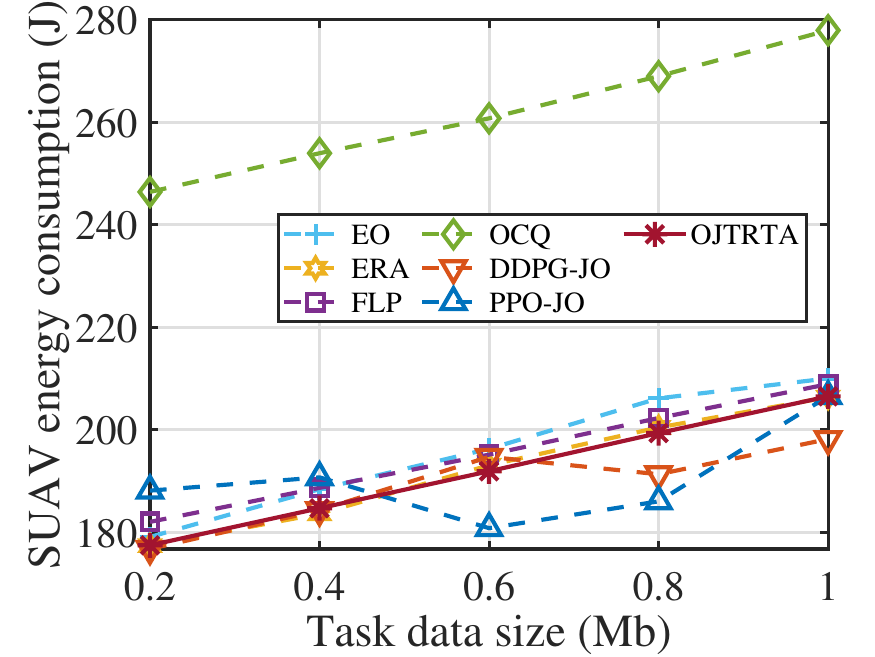}	
		\end{minipage}
	}
	\centering
	\caption{The impact of task data sizes on system performance. (a) Time-averaged UD cost. (b) Average task completion latency. (c) Cumulative UD energy consumption. (d) Time-averaged SUAV energy consumption.}
	\label{fig_tasksize}
	\vspace{-1em}
\end{figure*}

\begin{figure*}[!hbt] 
	\centering
	\setlength{\abovecaptionskip}{2pt}%
	\setlength{\belowcaptionskip}{2pt}%
	\subfigure[]
	{
		\begin{minipage}[t]{0.23\linewidth}
			\centering
			\includegraphics[scale=0.31]{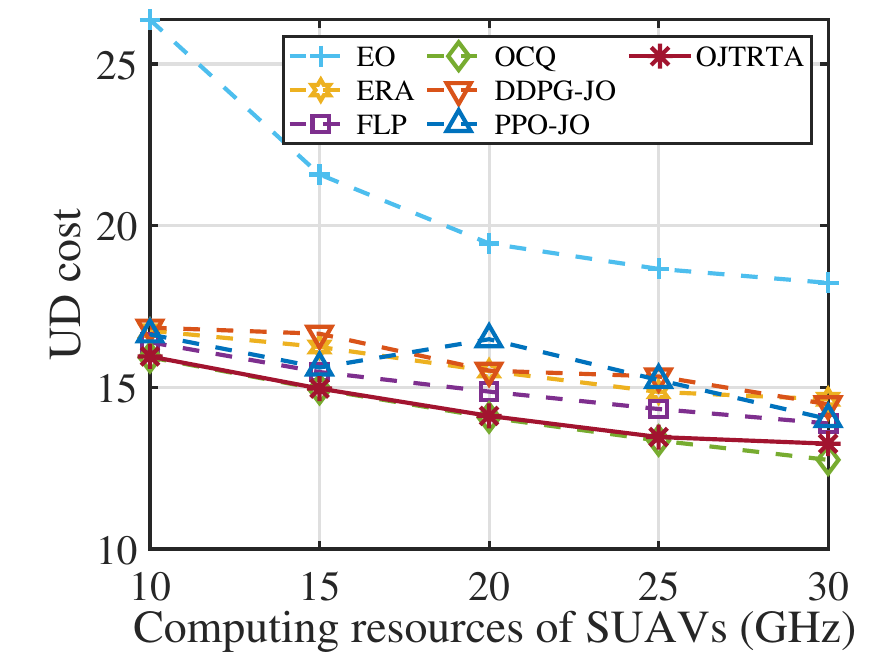}
		\end{minipage}
	}
	\subfigure[]
	{
		\begin{minipage}[t]{0.23\linewidth}
			\centering
			\includegraphics[scale=0.31]{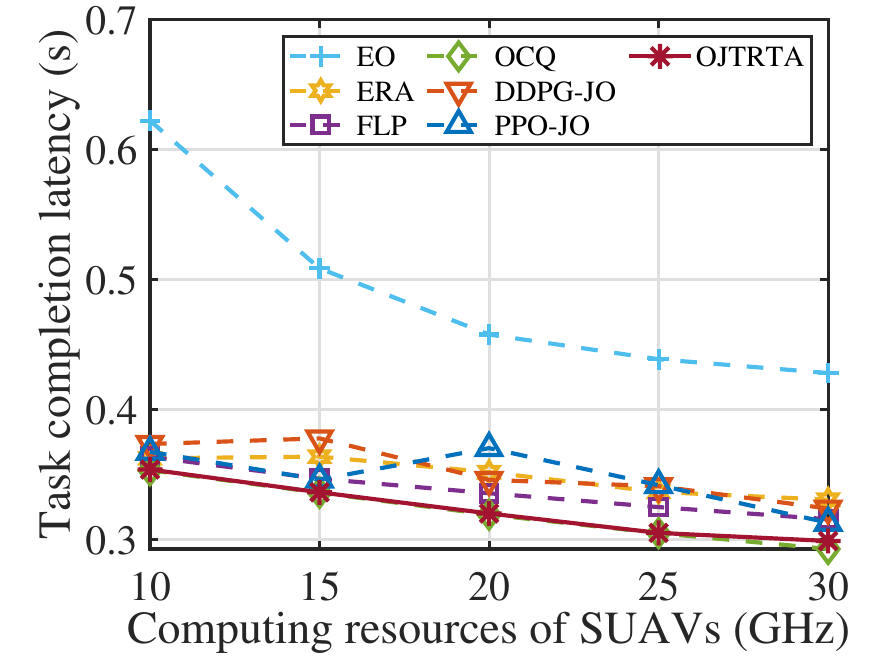}	
		\end{minipage}
	}
	\subfigure[]
	{
		\begin{minipage}[t]{0.23\linewidth}
			\centering
			\includegraphics[scale=0.31]{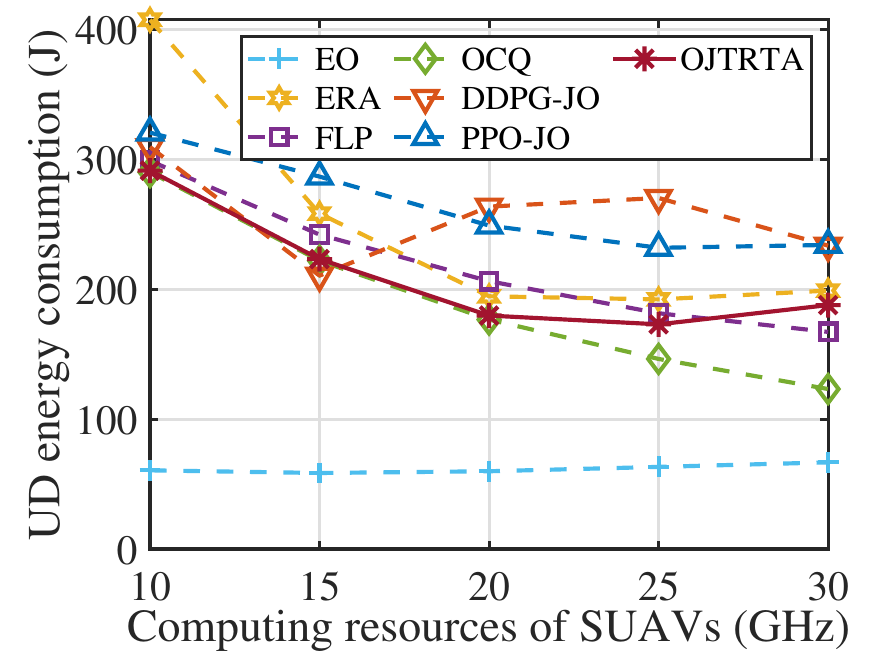}	
		\end{minipage}
	}
	\subfigure[]
	{
		\begin{minipage}[t]{0.23\linewidth}
			\centering
			\includegraphics[scale=0.31]{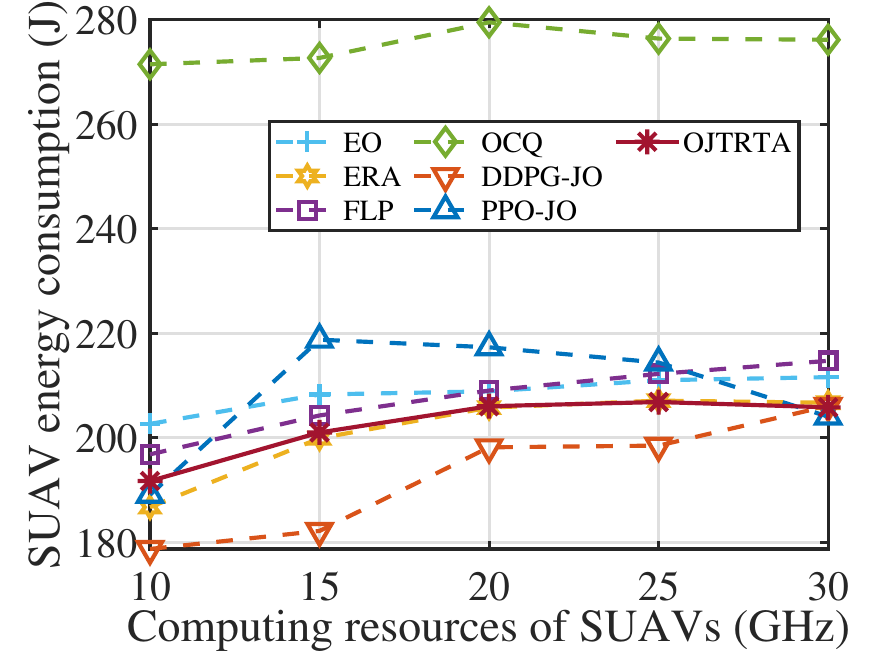}	
		\end{minipage}
	}
	\centering
	\caption{The impact of SUAV computing resources on system performance. (a) Time-averaged UD cost. (b) Average task completion latency. (c) Cumulative UD energy consumption. (d) Time-averaged SUAV energy consumption.}
	\label{fig_MEC}
	\vspace{-1em}
\end{figure*}

\subsubsection{Impact of UD Numbers}
\label{subsubsec:Impact of Parameters}

\par Figs. \ref{fig_UDnum}(a), \ref{fig_UDnum}(b), \ref{fig_UDnum}(c), and \ref{fig_UDnum}(d) illustrate the impact of varying numbers of UDs on the time-averaged UD cost, average task completion delay, cumulative UD energy consumption, and time-averaged SUAV energy consumption, respectively. It can be observed that the time-averaged UD cost, cumulative UD energy consumption and time-averaged SUAV energy consumption of all approaches show an upward trend as the number of UDs increases, since more computation tasks need to be processed. Furthermore, as the number of UDs increases, EO exhibits the poorest performance in terms of average task completion latency, while demonstrating optimal performance in cumulative UD energy consumption. This is primarily due to EO offloading all tasks to edge servers, resulting in lower UD energy consumption but higher server loads. Additionally, compared to the proposed approach, FLP shows inferior performance in terms of average task completion latency and cumulative UD energy consumption. This highlights the importance of optimizing the trajectory of UAVs. Besides, when the number of UDs reaches 100, OCQ and PPO-JO slightly outperform the proposed approach in terms of time-averaged UD cost. This can be attributed to the fact that OCQ and PPO-JO do not consider SUAV energy consumption constraints. As shown in Fig. \ref{fig_UDnum}(d), OCQ and PPO-JO exhibit higher time-averaged SUAV energy consumption.

\par Lastly, the proposed OJTRTA approach demonstrates superior performance in terms of time-averaged UD cost and average task completion latency compared to other approaches, while also exhibiting favorable performance in terms of cumulative UD energy consumption and time-averaged SUAV energy consumption. Specifically, compared with EO, ERA, FLP, DDPG-JO, and PPO-JO, the proposed approach can respectively reduce the time-averaged UD cost by 55.6\%, 3.2\%, 2.1\%, 9.4\%, 14.6\% and can respectively reduce the average task completion latency 58.7\%, 1.2\%, 3.1\%, 4.5\%, 15.4\% in the relative dense scenario ($M = 100$). In conclusion, the set of simulation results indicates that the proposed OJTRTA has better scalability with an increasing number of UDs.

\subsubsection{Impact of Task Data Size}

\par Figs. \ref{fig_tasksize}(a), \ref{fig_tasksize}(b), \ref{fig_tasksize}(c), and \ref{fig_tasksize}(d) show the impact of task data sizes on the time-averaged UD cost, average task completion latency, cumulative UD energy consumption, and time-averaged SUAV energy consumption. It can be seen that with the increase in task data sizes, there is an upward trend in terms of the time-averaged UD cost, average task completion latency, cumulative UD energy consumption, and time-averaged SUAV energy consumption. This is expected as the larger task data size leads to higher overheads on computing, communication, and energy consumption for UDs and UAVs. Moreover, EO exhibits a significant growth trend in terms of the time-averaged UD cost and average task completion latency compared to other approaches. This is due to the increased competition for the limited computational and communication resources of aerial servers caused by the entire offloading of EO as the task data size increases.

\par Finally, compared with EO, ERA, FLP, DDPG-JO and PPO-JO, OJTRTA achieves performance improvements of approximately 28.4\%, 9.7\%, 3.8\%, 9.8\%, and 7.2\% in terms of time-averaged UD cost, as well as 30.4\%, 9.3\%, 3.2\%, 8.5\%, and 3.9\% in terms of average task completion latency when the task data size reaches 1 Mb. The set of simulation results indicates that the proposed OJTRTA is able to adapt to the heavy-loaded scenarios with overall superior performances. 

\subsubsection{Impact of Computing Resources of SUAVs}

\par Figs. \ref{fig_MEC}(a), \ref{fig_MEC}(b), \ref{fig_MEC}(c), and \ref{fig_MEC}(d) depict the impact of different SUAV computing resources on the time-averaged UD cost, average task completion delay, cumulative UD energy consumption, and time-averaged SUAV energy consumption. It can be observed that with the increase of SUAV computing resources, all approaches show a decreasing trend in terms of the time-averaged UD cost and average task completion latency, and EO and OJTRTA demonstrate gradually decreasing performance improvements. The reasons can be explained as follows. The increase of SUAV computational resources provides more computing resource allocation for task execution, reducing the task execution latency. However, as SUAV computing resources further increase, communication resources and the energy constraints of SUAV become bottlenecks that limit the improvement of system performance. Furthermore, EO maintains nearly constant performance in terms of the cumulative UD energy consumption regardless of the variations in the computing resources of SUAVs. This is mainly because the entire offloading of EO primarily incurs transmission energy consumption for UDs, which is independent of the computing resources of SUAVs. 

\par Finally, OJTRTA outperforms EO, ERA, FLP, DDPG-JO and PPO-JO in terms of the time-averaged UD cost and average task completion delay, which illustrates the proposed approach enables sustainable utilization of computing resources and prevents resource over-utilization.

\subsubsection{UAV Trajectory Control}
\label{subsubsec:UAV Trajectory Planning}

\par Fig.~\ref{fig.trajectory} shows the mobility of UDs and the flight trajectories of SUAVs controlled by the proposed OJTRTA approach. It can be observed that the flight trajectories of SUAVs are in accordance with intuition. Specifically, SUAVs tend to fly towards areas with dense UDs to serve a greater number of UDs and provide improved communication conditions, thereby enhancing the QoE of the serviced UDs. After reaching the dense UD area, SUAVs adjust their positions to accommodate the dynamic computing demands of the UDs. Besides, it is worth noting that the proposed approach incorporates collision avoidance, effectively ensuring the safe flight of SUAVs in dense UD areas. In conclusion, the simulation results demonstrate that the trajectory control of the proposed OJTRTA approach can effectively regulate the trajectories of SUAVs according to the dynamic computing requirements to improve the QoE of UDs.

\begin{figure}[!hbt]
	\setlength{\abovecaptionskip}{0pt}%
	\setlength{\belowcaptionskip}{0pt}%
	\centering
	\includegraphics[width =3.5in]{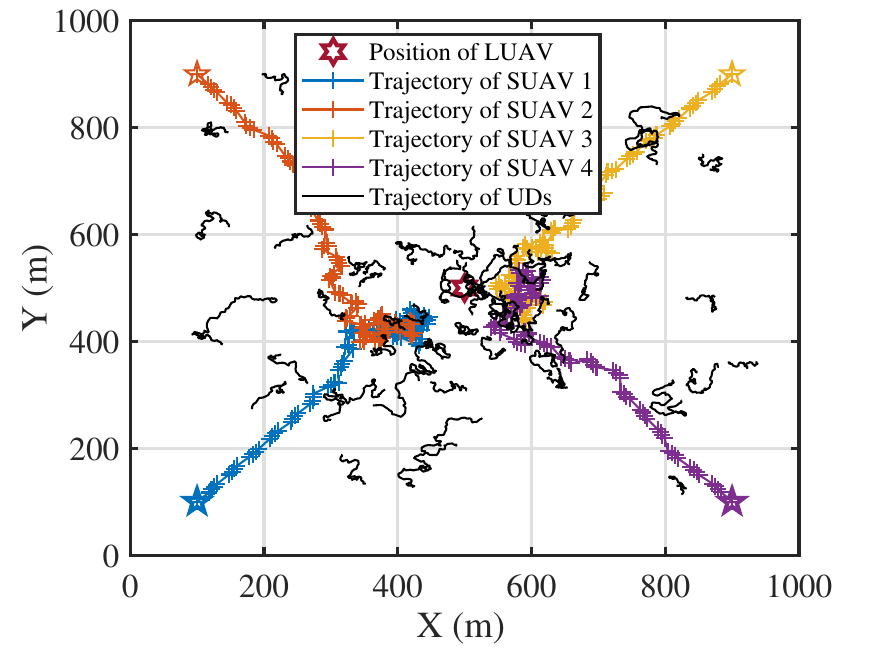}
	\caption{Trajectories of UDs and SUAVs.}
	\label{fig.trajectory}
\end{figure}

\section{Conclusion}
\label{sec:Conclusion}

\par In this work, we investigated the task offloading, resource allocation, and UAV trajectory control within a multi-UAV-assisted MEC system. A JTRTOP was formulated to maximize the QoE of all UDs while satisfying the energy and resource constraints of UAVs. Since the JTRTOP is future-dependent and NP-hard, we proposed the OJTRTA to solve the problem. Specifically, the future-dependent JTRTOP was first transformed into the PROP by using Lyapunov optimization methods. Furthermore, a two-stage optimization method was proposed to solve the PROP. Simulation results show that OJTRTA outperforms the comparative approaches in terms of time-averaged UD cost while meeting the SUAV energy consumption constraint. Furthermore, JTRTOP exhibits superior adaptability in heavy-loaded scenarios and demonstrates good scalability as the number of MDs increases. In the future, our efforts will focus on extending current research to space-air-ground integrated MEC networks. Specifically, the proposed multi-UAV-assisted MEC system heavily relies on the limited computational resources and energy supply of UAVs, which remains a critical bottleneck restricting the improvement of system performance. With the rapid deployment of low Earth orbit satellites, space-air-ground integrated networks demonstrate significant potential to substantially expand available network resources and the service coverage of MEC by aggregating the capabilities of heterogeneous network.

\bibliographystyle{IEEEtran}
\bibliography{ref}

\newpage

\begin{IEEEbiography}[{\includegraphics[width=1in,height=1.25in,clip,keepaspectratio]{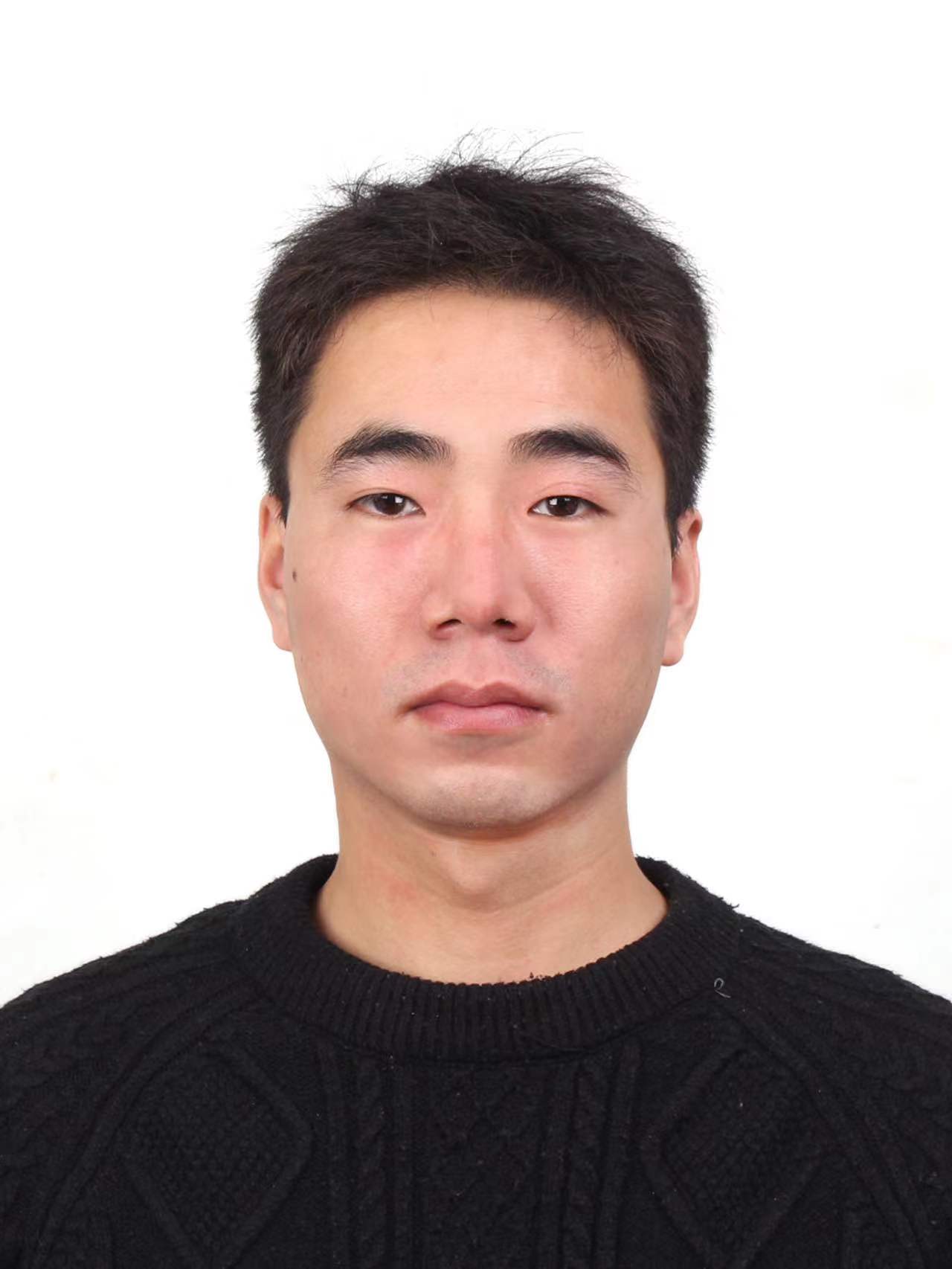}}]{Long He} received a BS degree in Computer Science and Technology from Chengdu University of Technology, Sichuan, China, in 2019. He is currently working toward the PhD degree in Computer Science and Technology at Jilin University, Changchun, China. His research interests include vehicular networks and edge computing.
\end{IEEEbiography}

\begin{IEEEbiography}[{\includegraphics[width=1in,height=1.25in,clip,keepaspectratio]{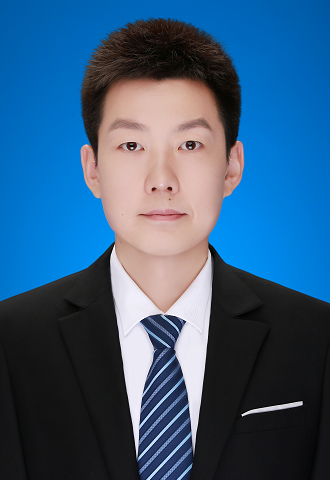}}]{Geng Sun} (Senior Member, IEEE) received the B.S. degree in communication engineering from Dalian Polytechnic University, and the Ph.D. degree in computer science and technology from Jilin University, in 2011 and 2018, respectively. He was a Visiting Researcher with the School of Electrical and Computer Engineering, Georgia Institute of Technology, USA. He is an Associate Professor in College of Computer Science and Technology at Jilin University, and his research interests include wireless networks, UAV communications, collaborative beamforming and optimizations.
\end{IEEEbiography}

\begin{IEEEbiography}[{\includegraphics[width=1in,height=1.25in,clip,keepaspectratio]{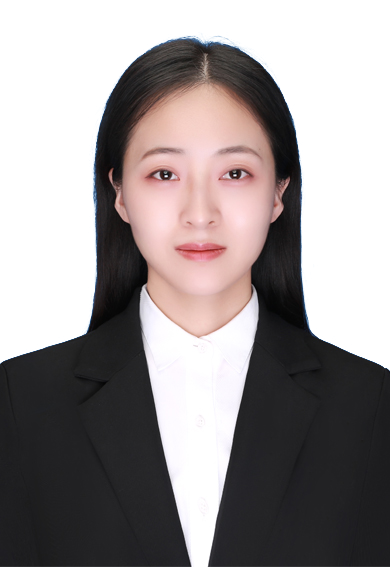}}]{Zemin Sun} received a BS degree in Software Engineering, an MS degree and a Ph.D degree in Computer Science and Technology from Jilin University, Changchun, China, in 2015, 2018, and 2022, respectively. Her research interests include vehicular networks, edge computing, and game theory. 
\end{IEEEbiography}

\begin{IEEEbiography}[{\includegraphics[width=1in,height=1.25in,clip,keepaspectratio]{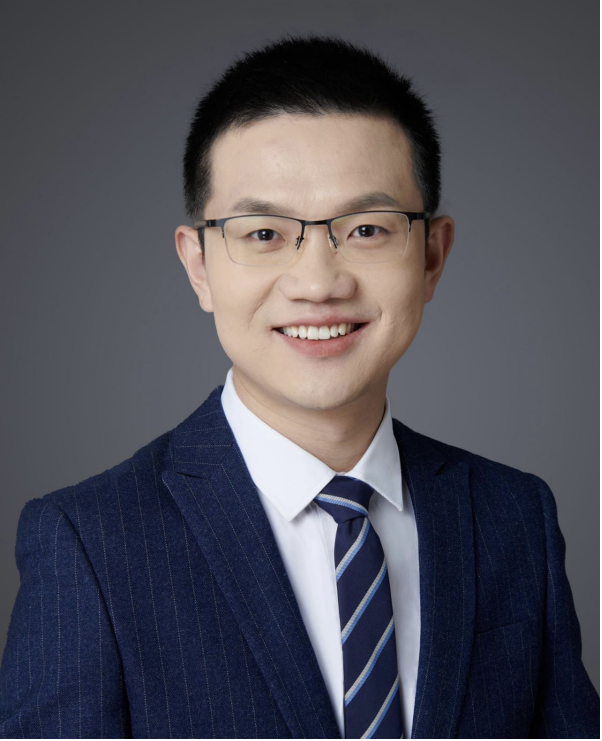}}]{Qingqing Wu} (Senior Member, IEEE) received the B.Eng. and the Ph.D. degrees in Electronic Engineering from South China University of Technology and Shanghai Jiao Tong University (SJTU) in 2012 and 2016, respectively. From 2016 to 2020, he was a Research Fellow in the Department of Electrical and Computer Engineering at National University of Singapore. He is currently an Associate Professor with Shanghai Jiao Tong University. His current research interest includes intelligent reflecting surface (IRS), unmanned aerial vehicle (UAV) communications, and MIMO transceiver design. He has coauthored more than 100 IEEE journal papers with 26 ESI highly cited papers and 8 ESI hot papers, which have received more than 18,000 Google citations. He was listed as the Clarivate ESI Highly Cited Researcher in 2022 and 2021, the Most Influential Scholar Award in AI-2000 by Aminer in 2021 and World’s Top 2\% Scientist by Stanford University in 2020 and 2021.
\end{IEEEbiography}

\begin{IEEEbiography}[{\includegraphics[width=1in,height=1.25in,clip,keepaspectratio]{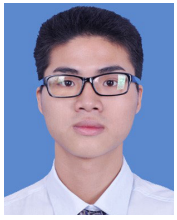}}]{Jiawen Kang} (Senior Member, IEEE) received the Ph.D. degree from Guangdong University of Technology, China, in 2018. He was a Post-Doctoral Researcher with Nanyang Technological University, Singapore, from 2018 to 2021. He is currently a Professor with Guangdong University of Technology. His main research interests include blockchain, security, and privacy protection in wireless communications and networking.
\end{IEEEbiography}

\begin{IEEEbiography}[{\includegraphics[width=1in,height=1.25in,clip,keepaspectratio]{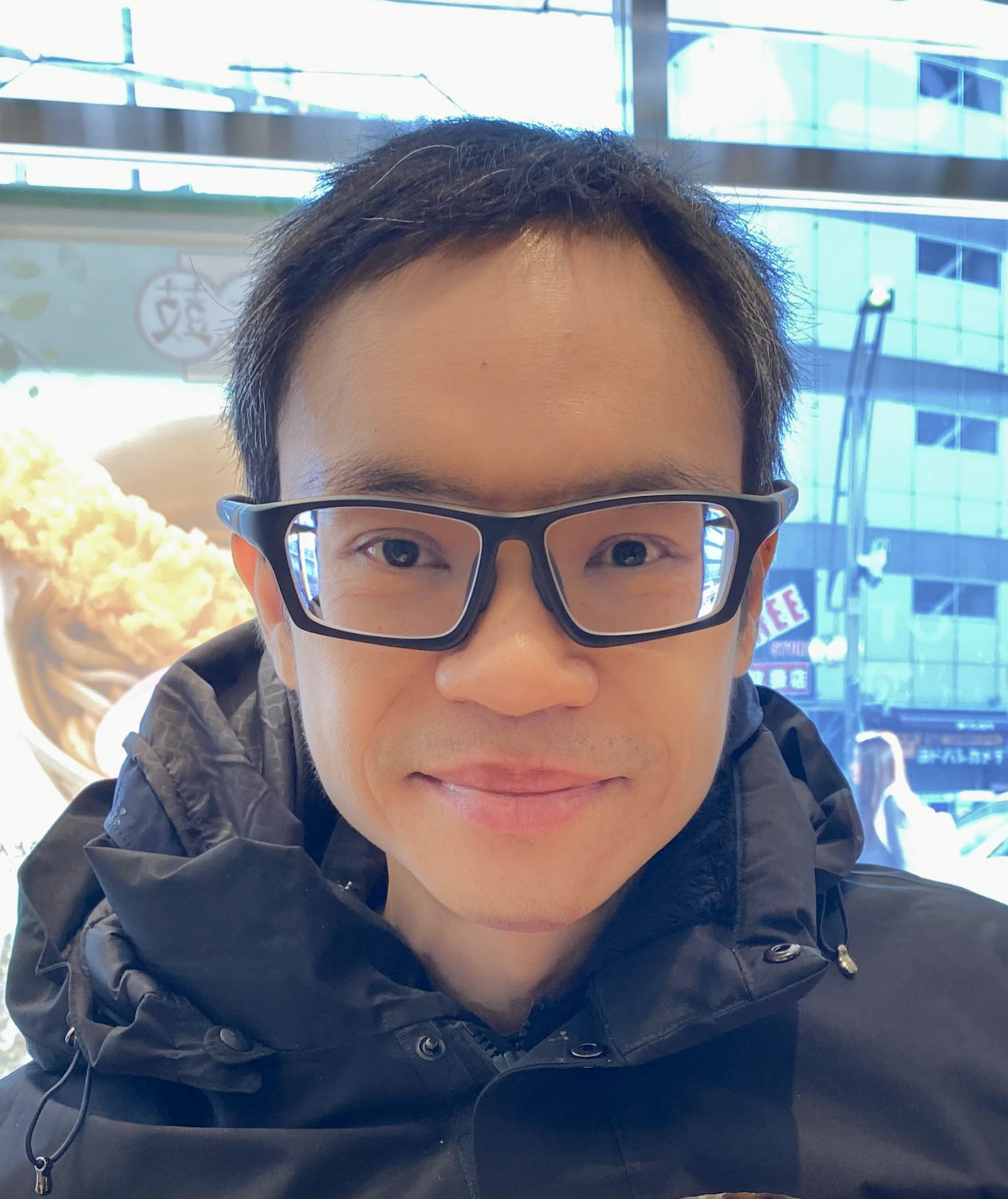}}]{Dusit Niyato} (Fellow, IEEE) received the B.Eng. degree from the King Mongkuts Institute of Technology Ladkrabang (KMITL), Thailand, in 1999, and the Ph.D. degree in electrical and computer engineering from the University of Manitoba, Canada, in 2008. He is currently a Professor with the School of Computer Science and Engineering, Nanyang Technological University, Singapore. His research interests include the Internet of Things (IoT), machine learning, and incentive mechanism design.
\end{IEEEbiography}

\begin{IEEEbiography}[{\includegraphics[width=1in,height=1.25in,clip,keepaspectratio]{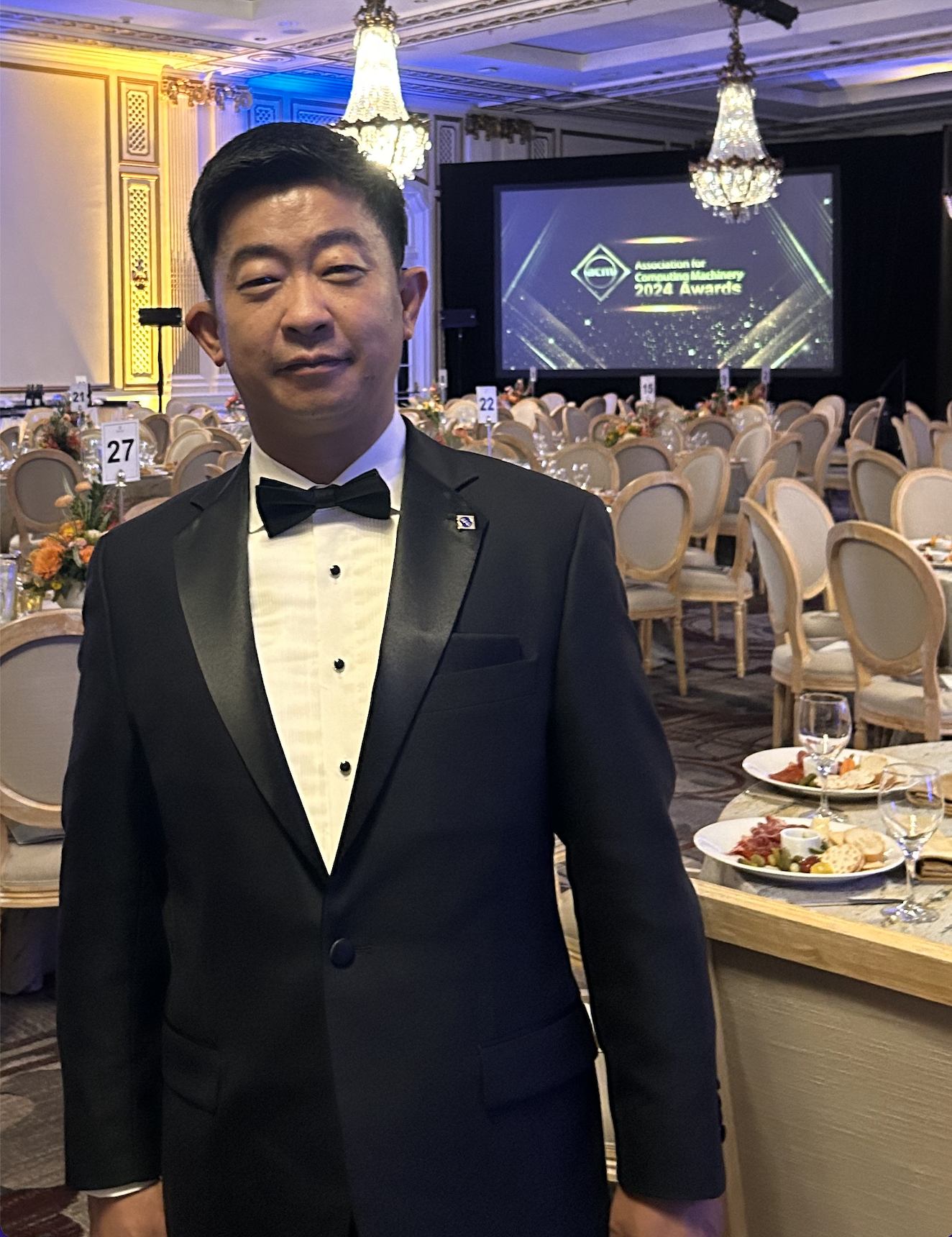}}]{Zhu Han} (S’01–M’04-SM’09-F’14) received the B.S. degree in electronic engineering from Tsinghua University, in 1997, and the M.S. and Ph.D. degrees in electrical and computer engineering from the University of Maryland, College Park, in 1999 and 2003, respectively. From 2000 to 2002, he was an R\&D Engineer of JDSU, Germantown, Maryland. From 2003 to 2006, he was a Research Associate at the University of Maryland. From 2006 to 2008, he was an assistant professor at Boise State University, Idaho. Currently, he is a John and Rebecca Moores Professor in the Electrical and Computer Engineering Department as well as in the Computer Science Department at the University of Houston, Texas. Dr. Han’s main research targets on the novel game-theory related concepts critical to enabling efficient and distributive use of wireless networks with limited resources. His other research interests include wireless resource allocation and management, wireless communications and networking, quantum computing, data science, smart grid, carbon neutralization, security and privacy.  Dr. Han received an NSF Career Award in 2010, the Fred W. Ellersick Prize of the IEEE Communication Society in 2011, the EURASIP Best Paper Award for the Journal on Advances in Signal Processing in 2015, IEEE Leonard G. Abraham Prize in the field of Communications Systems (best paper award in IEEE JSAC) in 2016, IEEE Vehicular Technology Society 2022 Best Land Transportation Paper Award, and several best paper awards in IEEE conferences. Dr. Han was an IEEE Communications Society Distinguished Lecturer from 2015 to 2018 and ACM Distinguished Speaker from 2022 to 2025, AAAS fellow since 2019, and ACM Fellow since 2024. Dr. Han is a 1\% highly cited researcher since 2017 according to Web of Science. Dr. Han is also the winner of the 2021 IEEE Kiyo Tomiyasu Award (an IEEE Field Award), for outstanding early to mid-career contributions to technologies holding the promise of innovative applications, with the following citation: ``for contributions to game theory and distributed management of autonomous communication networks."
\end{IEEEbiography}

\begin{IEEEbiography}[{\includegraphics[width=1in,height=1.25in,clip,keepaspectratio]{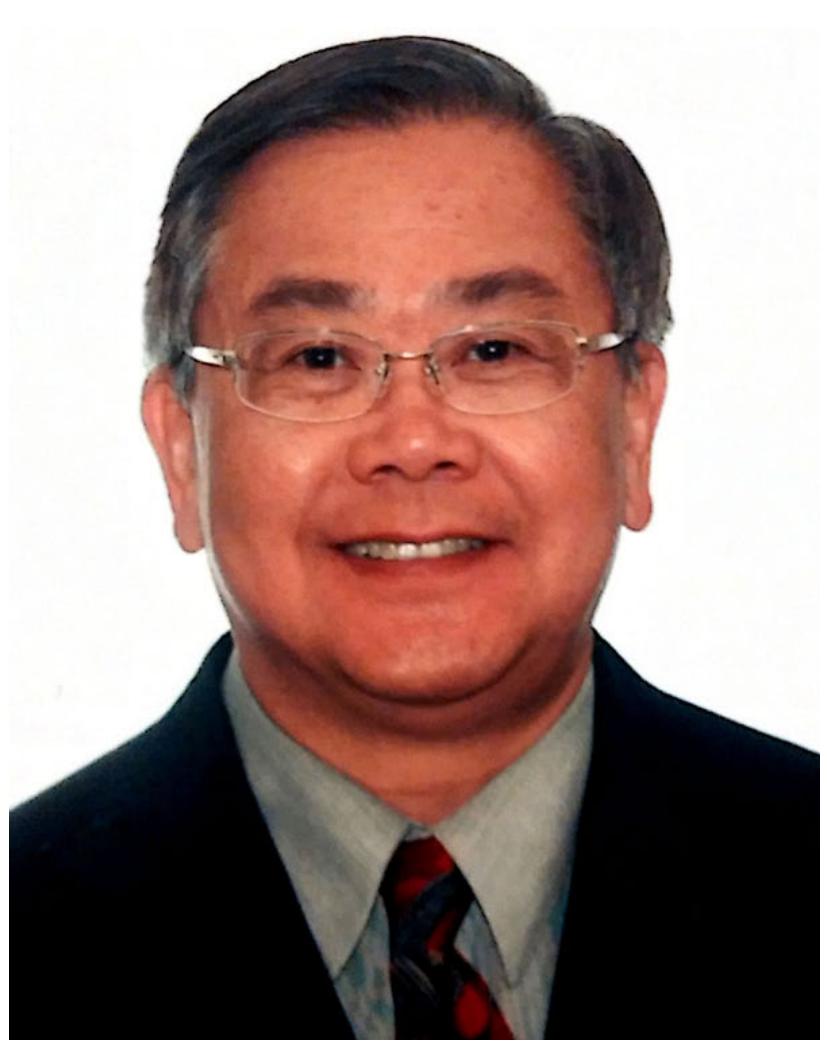}}]{Victor C. M. Leung} (Life Fellow, IEEE) is a Distinguished Professor of computer science and software engineering with Shenzhen University,
	China. He is also an Emeritus Professor of electrial and computer engineering and the Director of the Laboratory for Wireless Networks and Mobile Systems at the University of British Columbia (UBC). His research is in the broad areas of wireless networks and mobile systems. He has co-authored more than 1300 journal/conference papers and book chapters. Dr. Leung is serving on the editorial boards of IEEE Transactions on Green Communications and Networking, IEEE Transactions on Cloud Computing, IEEE Access, and several other journals. He received the IEEE Vancouver Section Centennial Award, 2011 UBC Killam Research Prize, 2017 Canadian Award for Telecommunications Research, and 2018 IEEE TCGCC Distinguished Technical Achievement Recognition Award. He co-authored papers that won the 2017 IEEE ComSoc Fred W. Ellersick Prize, 2017 IEEE Systems Journal Best Paper Award, 2018 IEEE CSIM Best Journal Paper Award, and 2019 IEEE TCGCC Best Journal Paper Award. He is a Life Fellow of IEEE, and a Fellow of the Royal Society of Canada, Canadian Academy of Engineering, and Engineering Institute of Canada. He is named in the current Clarivate Analytics list of Highly Cited Researchers.
\end{IEEEbiography}
\end{document}